\documentclass[reqno]{amsart}
\usepackage{amssymb}
\usepackage[english]{babel}
\usepackage{xspace}
\usepackage[bookmarksnumbered,colorlinks]{hyperref}
\usepackage{geometry}
\usepackage{enumitem}
\usepackage{mathrsfs}
\renewcommand{\hat}{\widehat}

\newtheorem{theorem}{Theorem}[section]
\newtheorem{lemma}[theorem]{Lemma}
\newtheorem{definition}[theorem]{Definition}
\newtheorem{corollary}[theorem]{Corollary}
\newtheorem{proposition}[theorem]{Proposition}
\theoremstyle{remark}
\newtheorem{remark}[theorem]{Remark}
\theoremstyle{definition}
\newtheorem{example}[theorem]{Example}

\newcommand{\bt}{\begin{theorem}}
	\newcommand{\et}{\end{theorem}}
\newcommand{\bco}{\begin{corollary}}
	\newcommand{\eco}{\end{corollary}}
\newcommand{\bd}{\begin{definition}}
	\newcommand{\ed}{\end{definition}}
\newcommand{\bl}{\begin{lemma}}
	\newcommand{\el}{\end{lemma}}
\newcommand{\bpr}{\begin{proposition}}
	\newcommand{\epr}{\end{proposition}}
\newcommand{\bere}{\begin{remark}}
	\newcommand{\ere}{\end{remark}}
\newcommand{\beq}{\begin{equation}}
	\newcommand{\eeq}{\end{equation}}

\newcommand{\R}{\ensuremath{\mathbb{R}}\xspace}

\newcommand{\scrS}{\mathscr{S}}
\newcommand{\scrT}{\mathscr{T}}
\renewcommand{\H}{\ensuremath{\mathcal{H}}\xspace}
\newcommand{\A}{\ensuremath{\mathcal{A}}\xspace}

\newcommand{\Lie}{\ensuremath{\mathcal{L}}\xspace}

\title[Photon and massive particle hypersurfaces]{A unified framework for photon and massive particle hypersurfaces in stationary spacetimes}

\author[E. Caponio]{Erasmo Caponio}
\address{Dipartimento di Meccanica, Matematica e Management\hfill\break\indent
	Politecnico di Bari \hfill\break\indent Via Orabona 4,
	70125, Bari, Italy}
\email{erasmo.caponio@poliba.it}

\author[A.V. Germinario]{Anna Valeria Germinario}
\address{Dipartimento  di Matematica\hfill\break\indent
	Universit\`a degli Studi di Bari Aldo Moro\hfill\break\indent Via  Orabona 4,
	70125 Bari Italy}
\email{anna.germinario@uniba.it}
\author[A. Masiello]{Antonio Masiello}
\address{Dipartimento di Meccanica, Matematica e Management\hfill\break\indent
	Politecnico di Bari \hfill\break\indent Via Orabona 4,
	70125, Bari, Italy}
\email{antonio.masiello@poliba.it}

\date{}

\begin{document}
	
\begin{abstract}
Massive particle surfaces (MPS), introduced by Kobialko, Bogush and Gal'tsov,
are timelike hypersurfaces confining the trajectories of charged massive
particles with fixed Killing energy and charge-to-mass ratio, in analogy with
the photon surfaces of Claudel, Virbhadra and Ellis. We study both classes in
standard stationary spacetimes within a unified Finslerian framework. For
Killing-invariant hypersurfaces $\mathbb{R}\times S_0$, with $S_0$ a
hypersurface of a spacelike slice $S$, we prove that photon hypersurfaces and
MPS are exactly those for which $S_0$ is totally geodesic with respect to the
Fermat metric and to a Jacobi--Randers metric on $S$, respectively. For MPS we
also give a gauge-invariant version, in terms of a Riemannian Jacobi metric
and a tangency condition on the reduced magnetic field. For arbitrary timelike
hypersurfaces, we give a coordinate-free formulation  of the partial
umbilicity characterization of MPS and, in dimension at least four and under
a natural energy condition, we characterize when an MPS is quasi-umbilical in
the sense of B.-Y. Chen. In the Som--Raychaudhuri spacetime of a rigidly
rotating charged dust we exhibit MPS which are neither photon hypersurfaces
nor quasi-umbilical. Finally, under completeness assumptions on the
Jacobi--Randers metric, we prove existence and multiplicity results for
proper-time parametrized solutions of the Lorentz force equation connecting
a point to a flow line of the Killing field and, under compactness
assumptions, existence of solutions with periodic spatial projection,
possibly confined to an MPS.	
\end{abstract}

\keywords{Photon surfaces, massive particle surfaces,  stationary spacetimes, Lorentz force equation, Randers metrics}
\subjclass[2020]{53C50, 53C60,  70H33, 83C50}
\maketitle
	
\section{Introduction}
The study of photon surfaces has attracted renewed interest following the observations of black hole shadows by the Event Horizon Telescope \cite{EHT2019}, which are closely connected with the behavior of null geodesics. A timelike photon hypersurface is a hypersurface such that every null geodesic initially tangent to it locally remains in it; Claudel, Virbhadra, and Ellis \cite{Claudel2001} characterized these hypersurfaces by total umbilicity, relating null geodesic dynamics to extrinsic geometry. The motion of massive, possibly charged, particles around compact objects also carries information on the geometry of spacetime. This motivates the notion of a massive particle surface (MPS) introduced in \cite{KBG2022}: a timelike hypersurface that confines the trajectories of massive particles with prescribed energy and charge. Its extrinsic characterization is a partial umbilicity condition involving the electromagnetic field.

Our aim is to describe photon and massive particle hypersurfaces in stationary spacetimes through geometric structures on a spacelike slice. We work with Killing-invariant hypersurfaces $\scrT=\R\times S_0$ in a standard stationary splitting $M=\R\times S$; such a splitting always exists locally, and globally under suitable causality and completeness assumptions \cite{JavSan2008}. For massive particles, we fix the charge-to-mass ratio $\rho$ and the Killing energy $\varepsilon$, and use proper-time parametrization; we call the corresponding hypersurfaces $(\rho,\varepsilon)$--MPS.

Our main results characterize both classes by the total geodesicity of $S_0$ with respect to suitable Finsler metrics on $S$. For photon hypersurfaces, the relevant metric is the Fermat metric, whose pregeodesics, i.e.\ geodesics up to an orientation-preserving reparametrization, are the projections of future-pointing null geodesics \cite{cjm,cjs}: $\scrT$ is a photon hypersurface if and only if $S_0$ is totally geodesic for this metric (Theorem~\ref{thm:photonumb}). For massive particles, Theorem~\ref{lorentz-lag} reduces the Lorentz force equation to a non-relativistic electromagnetic system on the slice, whose mechanical energy equals $e=\tfrac12 g(\dot\gamma,\dot\gamma)$, so that proper-time parametrization corresponds to $e=-\tfrac12$. A second reduction, at fixed mechanical energy, identifies the trajectories of this system with pregeodesics of a Jacobi--Randers metric, wherever the Randers condition holds. Consequently, if the Jacobi--Randers metric is defined on a neighborhood of $S_0$, then $\scrT$ is a $(\rho,\varepsilon)$--MPS if and only if $S_0$ is totally geodesic for it (Theorem~\ref{mpschar}). Since the Randers condition depends on the choice of electromagnetic gauge, we also prove a characterization that is invariant under stationary gauge transformations: assuming only that the Jacobi metric is positive definite near $S_0$, the MPS property is equivalent to total geodesicity of $S_0$ for the Jacobi metric together with tangency of the reduced magnetic force to $S_0$ (Theorem~\ref{riemannian-magnetic-MPS}).

The Jacobi--Randers reduction is related to \cite{Perlick1991,ch2019} for timelike geodesics and to \cite{LiJia2021} for charged particles, while more general relations between fixed-energy dynamics and Finsler geometry appear in \cite{Wein78,CIPP,mest2016}. In our electromagnetic setting, we give a direct proof of the correspondence (Appendix~\ref{app}), extending the argument of \cite{ben84}; this complements the approach through homogenization and Routh reduction in \cite{mest2016}, with the homogenization procedure described in \cite[Ch.~II, \S7]{Rund}. A related approach based on the curvature of a Riemannian Jacobi metric is developed in \cite{BerCardAnd25}.

The characterizations above rely on the effective dynamics on the slice, rather than on the extrinsic geometry of $\scrT$. In Section~\ref{secMPS} we revisit the extrinsic characterization of MPS given in \cite[Thm.~3.1]{KBG2022}, where it is stated for an arbitrary Killing vector field whose tangential projection may be timelike, spacelike or null (the null case being treated in \cite[App.~A]{KBG2022}). Restricting to timelike Killing fields, as in the rest of the paper, we give a coordinate-free formulation and proof of it, including a complete proof of the converse implication, and recover its scalar master equation (Theorem~\ref{thmumbMPS} and Proposition~\ref{master}). In spacetime dimension $d\geq4$, and under the energy condition $\mathcal E_k^2>\kappa^2$ (which ensures that the admissible velocities at each point form a nondegenerate sphere, see Remark~\ref{degenerate}), we characterize when an MPS is quasi-umbilical in the sense specified in Subsection~\ref{qumb}: this holds if and only if either $\rho=0$ or the normal contraction of the electromagnetic field is proportional to the metric dual of the tangential Killing field (Proposition~\ref{quasiumbilicequiv}).

We provide examples of MPS in the critical Som--Raychaudhuri charged-dust spacetime. Cancellation of the reduced magnetic field produces a family of $(2\varepsilon,\varepsilon)$--MPS, with $\varepsilon^2>1$, including hypersurfaces that are neither photon hypersurfaces nor quasi-umbilical (Example~\ref{SomRay}). The same spacetime also illustrates the threshold energy level (Remark~\ref{SomRay-remarks}).

Finally, where the Jacobi--Randers reduction applies, variational methods for Randers metrics yield global results (Theorem~\ref{pointline}): solutions of the Lorentz force equation with prescribed $\rho$, $\varepsilon$ and $e$ connecting a point to a flow line of the Killing field exist, infinitely many of them when the slice is not contractible, \eqref{dis} holds and their arrival times tend to $+\infty$; on a compact slice, there are trajectories with nonconstant periodic projection. We give a condition ensuring that the Randers domain (Definition~\ref{Randersdomain}) is the whole slice for large $|\varepsilon|$ (Proposition~\ref{global-high-energy}), so that on a compact slice such trajectories exist, parametrized by proper time, for every $\rho$ and every sufficiently large $|\varepsilon|$ (Corollary~\ref{closedlarge}). Moreover, if $\scrT=\R\times S_0$ is a $(\rho,\varepsilon)$--MPS with $S_0$ compact and contained in the Randers domain, there are proper-time parametrized trajectories confined to $\scrT$ with nonconstant periodic projection on $S_0$ (Corollary~\ref{closedMPS}).

Throughout the paper, the admissible initial data of massive particles are singled out by fixing $\rho$ and the Killing energy $\varepsilon$ only. Fixing further conserved quantities, e.g.\ also the angular momentum in axisymmetric spacetimes, weakens the confinement requirement and enlarges the class of massive particle hypersurfaces, including, for instance, the hypersurfaces of spherical orbits outside the ergoregion of Kerr--Newman spacetime; see \cite{BKG2023} for the related notion of glued massive particle surfaces. A Finslerian description of these more general settings would require a further reduction with respect to the second Killing field and is left to future work.

The paper is organized as follows. Section~\ref{stationaryFermat} recalls the standard stationary splitting, the Fermat principle in this class of spacetime and the Fermat metrics.  Section~\ref{photonsurf} treats photon hypersurfaces. Section~\ref{secMPS} introduces massive particle hypersurfaces and studies their extrinsic geometry. Section~\ref{KKJM} develops the electromagnetic and Jacobi--Randers reductions, proves the two intrinsic characterizations, and presents the Som--Raychaudhuri example and the global existence results. Appendix~\ref{app} proves the correspondence between solutions with fixed total energy of a non-relativistic electromagnetic system and  pregeodesics of a Jacobi--Randers metric.

	\section{Stationary spacetimes: standard form and Fermat principles}\label{stationaryFermat}
	Throughout we assume that $(M,g)$ is a (connected) spacetime of dimension $d\geq 3$ admitting a complete, future-pointing, timelike, conformal Killing field $K$. 
	By \cite{JavSan2008}, $(M,g)$ splits globally as a \emph{standard conformastationary} spacetime with respect to $K$ 	if and only if it is distinguishing. More precisely,  under the (optimal) hypotheses $K$ complete and $(M,g)$ distinguishing,	there is a diffeomorphism $M\cong \R\times S$ such that, in the induced coordinates $(t,x)\in\R\times S$ with $K=\partial_t$, $g_{(t,x)}=\varphi(t,x)\Big(-\beta(x)\,dt^2+2\,\omega_x(\cdot)\,dt+g_0|_x\Big)$,
	where $\varphi>0$, $\beta>0$, $\omega$ is a one-form on $S$ and $g_0$ is a Riemannian metric on $S$. 
	Moreover, replacing $g$ by the conformal metric $g^*=g/\varphi$ one may assume $K$ is Killing; i.e. we work
	with the \emph{standard stationary} representative
$g=-\beta(x)\,dt^2+2\,\omega_x(\cdot)\,dt+g_0|_x$.
	This normalization preserves photon surfaces because null pregeodesics are conformally invariant. For massive particle surfaces, we restrict throughout to stationary metrics of the above standard form.
	
	Introduce the usual optical data
	\begin{equation}\label{opticaldata}
		\widehat\omega:=\frac{\omega}{\beta},\qquad 
		h_0:=g_0+\frac{1}{\beta}\,\omega\otimes\omega, \qquad 
		\widehat h:=\frac{1}{\beta}h_0= \frac{1}{\beta}\,g_0+\widehat\omega\otimes\widehat\omega,
	\end{equation}
	so that the metric $g$ can be written as
	\begin{equation}\label{1a}
		g_{(t,x)}  =  - \beta(dt-\widehat\omega_x)\otimes (dt-\widehat\omega_x)+h_0|_x .
	\end{equation}
	Any future-/past-pointing lightlike curve $\gamma=(t,\sigma)$ satisfies the null constraint
	\begin{equation*}
		0=g(\dot\gamma,\dot\gamma)=-\beta\big(\dot t-\widehat\omega(\dot \sigma)\big)^2+h_0(\dot \sigma,\dot \sigma).
	\end{equation*}
	The {\em Fermat  metrics} on $S$ are
	\begin{equation}\label{Fermat}
		F^\pm(x,y)=\sqrt{\,\widehat h_x(y,y)\,}\ \pm\ \widehat\omega_x(y).
	\end{equation}
	Their geodesics are related to null geodesics of $(M,g)$ by the following Fermat principle:
	
	\begin{theorem}[{\cite[Thm. 4.1]{cjm}}]\label{thm:Fermat}
		Let $p=(t_0,x_0)$ and let $x_1\in S$. A future-pointing lightlike curve, $\gamma\colon [0,1]\to M$,  $\gamma(s)=\big (t(s),\sigma(s)\big)$, from $p$ to the timelike line  $\R\times\{x_1\}$ is an  affinely parametrized geodesic of $(M,g/\beta)$ (and hence, up to reparametrization, a lightlike geodesic of $g$) if and only if $\sigma$ is a pregeodesic of $F^+$ joining $x_0$ to $x_1$, parametrized with constant $\widehat h$-speed   and  $t$ is given by 
		\begin{equation*}
			t(s)=t(0)+\int_0^s F^+\big(\sigma(u),\dot \sigma(u)\big )du.
		\end{equation*}
		For past-pointing lightlike curves the analogous statement holds with  $F^-$  and 
		\[t(s)=t(0)-\int_0^s F^-\big(\sigma(u),\dot \sigma(u)\big )\,du.\]
	\end{theorem}
	For every future-/past pointing lightlike curve $\gamma=(t,\sigma)$, the null 	constraint gives $\dot t=\pm F^\pm(\sigma,\dot\sigma)$, so that the arrival time 	$t(1)-t(0)$ of $\gamma$ at the line $\R\times\{x_1\}$ equals up to a sign the $F^\pm$-length 	of $\sigma$. Hence, Theorem~\ref{thm:Fermat} states that light rays from $p$
	to $\R\times\{x_1\}$ are the critical points of the arrival time, which in a 	conformally standard stationary spacetime is,  up to a sign, the length functional of the Randers metrics 	$F^\pm$. The arrival time functional was introduced by Kovner 	\cite{Kov90} for arbitrary spacetimes and soon after studied by Perlick
	\cite{Per90i,Per90ii}.  A  different integral formula for  the arrival time functional  of light rays in conformally standard stationary spacetimes, with existence and multiplicity results obtained by variational
	methods, was given in \cite{FGM95}. The name Fermat metrics for $F^\pm$ was 	introduced in \cite{cjm}. For their applications  to  	geometric and physical problems see, e.g., \cite{cjm,cjs,CapGerSan12, CaMa03} and the recent reviews 	\cite{Masiello2021,CJ}.

	\section{Photon surfaces as totally geodesic surfaces}\label{photonsurf}
	Let us recall the notion of a photon surface in \cite{Claudel2001}. Let $\mathscr T$ be a smooth non-degenerate hypersurface in $(M,g)$ and let $g_{\scrT}$ be the metric induced by $g$ on $\scrT$. Given any local unit normal vector field $n$,   the \emph{second fundamental form} is defined by $\Pi(X,Y):=g(\nabla_X n,\,Y)$  for all $X,Y\in \Gamma(T\mathscr T)$.
	
	\begin{definition}\label{photonsurface}
		A smooth nowhere-spacelike immersed hypersurface
		$\scrS\subset M$ is called a
		\emph{photon surface} if null geodesics initially tangent
		to $\scrS$ locally remain in $\scrS$.
		More precisely, for every $p\in\scrS$ and every nonzero
		$v\in T_p\scrS$ with $g_p(v,v)=0$, there exists $\delta>0$
		such that the geodesic $\gamma$ with initial data
		$\gamma(0)=p$ and $\dot\gamma(0)=v$ satisfies $	\gamma((-\delta,\delta))\subset\scrS$.
	\end{definition}
	If a photon surface $\scrS$ is embedded (without boundary) and closed as a
	subset of $M$, every maximal null geodesic initially tangent
	to $\scrS$ remains in $\scrS$ throughout its domain.
	Indeed, $T\scrS$ is closed in $TM$, so local confinement 
	extends to the whole interval of existence by continuity
	and uniqueness of the geodesic initial value problem.
	
	Notice that Definition~\ref{photonsurface}, following \cite{Claudel2001},
	also includes null hypersurfaces, which are trivially photon surfaces.
	The characterization in terms of total umbilicity recalled below applies
	to timelike photon hypersurfaces. This is the case relevant for 
	Killing-invariant hypersurfaces considered in the sequel.

	\begin{definition}[Killing-invariant hypersurface]
		Let $(M,g)$ be a spacetime with a complete Killing vector field $K$ and flow $\{\Phi_s\}_{s\in\mathbb{R}}$. A smooth, embedded, timelike hypersurface $\Sigma\subset M$ is \emph{$K$-invariant} if $\Phi_s(\Sigma)=\Sigma$ for all $s\in\mathbb{R}$. 
	\end{definition}
	Flow invariance implies that $K_p\in T_p\Sigma$ for every 	$p\in\Sigma$. Conversely, tangency implies flow invariance 	provided that $K|_\Sigma$ is complete; this holds, in particular, 	when $\Sigma$ is closed in $M$.

	\begin{remark}[Reduction under the standard stationary splitting]
		Under the hypotheses in \cite{JavSan2008} (complete $K$ and $(M,g)$ distinguishing), working in the standard stationary representative with $M\cong\mathbb{R}\times S$ and $K=\partial_t$, we see that any $K$-invariant  hypersurface is of the form $\Sigma=\mathbb{R}\times S_0$,
		where $S_0$ is a smooth codimension-one submanifold of $S$.
	\end{remark}
	
	\subsection{Characterization of photon surfaces}
	We recall the notions of umbilicity for a non-degenerate hypersurface in a spacetime $(M,g)$.
	
	\begin{definition}
		A non-degenerate hypersurface $\scrS$ is called \emph{totally umbilical} if $\Pi = H g_{\scrS}$ where $H := \frac{1}{d-1}g_{\scrS}^{\alpha\beta}\Pi_{\alpha\beta}$ is the mean curvature.
	\end{definition}
	
	It was proved in \cite{Claudel2001} that a timelike hypersurface
	$\scrT$ is a photon surface if and only if it is totally umbilical.
	In a standard stationary spacetime, for a $\partial_t$-invariant
	hypersurface $\scrT=\R\times S_0$, these conditions are also equivalent
	to $S_0$ being totally geodesic for either Fermat metric. 	Indeed, the Fermat metrics satisfy
	$F^-(x,y)=F^+(x,-y)$, 	so that $F^-$ is the reverse metric of $F^+$. The same relation holds
	for their restrictions $F^\pm_0$ to $TS_0$. Consequently, $S_0$ is 	totally geodesic with respect to $F^+$ if and only if it is totally 	geodesic with respect to $F^-$.

	Recall that given a Finsler manifold $(S, F)$,  an immersed submanifold  $S_0$ is totally geodesic for $F$ if, by definition, 	every geodesic of the induced Finsler metric $F_0$ on $S_0$ is also a geodesic of
	$F$.  Equivalently, every $F$-geodesic whose initial point lies
	in $S_0$ and whose initial velocity is tangent to $S_0$ remains locally
	in $S_0$.  	Indeed, assume first that every $F_0$-geodesic is an $F$-geodesic.
	Given an $F$-geodesic with initial data $(x,v)\in TS_0$, let $\sigma$
	be the $F_0$-geodesic with the same initial data. By assumption,
	$\sigma$ is also an $F$-geodesic, and uniqueness for the $F$-geodesic
	initial value problem implies that the two geodesics coincide locally.
	Hence the ambient geodesic remains in $S_0$. Conversely, assume that every $F$-geodesic with initial data in $TS_0$	remains locally in $S_0$. Let $\sigma$ be an $F_0$-geodesic and let
	$\gamma$ be the $F$-geodesic with the same initial data. By assumption,
	$\gamma$ is contained in $S_0$. Since an ambient $F$-geodesic contained
	in $S_0$ is also a geodesic of the induced metric $F_0$ (the first
	variation of the ambient energy vanishes, in particular, for variations
	through curves in $S_0$), uniqueness for the $F_0$-geodesic initial
	value problem gives $\gamma=\sigma$. Thus $\sigma$ is an $F$-geodesic.
	
	\begin{theorem}[$\partial_t$-invariant photon surfaces]
		\label{thm:photonumb}
		Let $\scrT=\R\times S_0$ be a $\partial_t$-invariant hypersurface
		in a standard stationary spacetime $(\R\times S,g)$. The following are
		equivalent:
		\begin{enumerate}
			\item $\scrT$ is a photon surface;
			\item $\scrT$ is totally umbilical;
			\item $S_0$ is totally geodesic with respect to $F^+$.
		\end{enumerate}
		Equivalently, $F^+$ in (3) may be replaced by $F^-$.
	\end{theorem}
	
	\begin{proof}
		The equivalence $(1)\Leftrightarrow(2)$ is
		\cite[Thm.~II.1]{Claudel2001}.
		
		Let us prove $(1)\Leftrightarrow(3)$. Suppose that $\scrT$ is a
		photon surface. Let $p\in S_0$ and
		$v\in T_pS_0\setminus\{0\}$. For any $t_0\in\R$, the vector
		$(F^+(v),v)\in T_{(t_0,p)}\scrT$ is future-pointing and null.
		By the Fermat principle (Theorem~\ref{thm:Fermat}), the
		corresponding null geodesic projects onto a pregeodesic of
		$F^+$ starting at $p$ with tangent direction $v$. Since
		$\scrT$ is a photon surface, the null geodesic remains in
		$\scrT$, and hence its projection remains in $S_0$.
		Thus $S_0$ is totally geodesic with respect to $F^+$.
		
		Conversely, suppose that $S_0$ is totally geodesic with respect
		to $F^+$. Let $\gamma$ be a future-pointing null geodesic tangent
		to $\scrT$. By Theorem~\ref{thm:Fermat}, its projection is a
		pregeodesic of $F^+$. Since it starts tangent to $S_0$, it remains
		in $S_0$, and therefore $\gamma$ remains in $\scrT$.
		The same conclusion holds for past-pointing null geodesics by
		reversing their parametrization. Hence $\scrT$ is a photon surface.
	\end{proof}
	\section{Massive particle surfaces at fixed charge-to-mass ratio and energy}\label{secMPS}
	Let $(M,g)$ be a stationary spacetime with Levi-Civita connection $\nabla$ and a future-pointing timelike Killing vector field $K$.
	Let $\mathcal A$ be a stationary electromagnetic potential ($\mathcal L_K \mathcal A=0$) with electromagnetic field $\mathcal F=d\mathcal A$.
	For a  timelike worldline $\gamma=\gamma(\tau)$ parametrized by proper time,  i.e.  $g(\dot\gamma,\dot\gamma)=-1$, let us fix  the \emph{charge-to-mass ratio} $\rho:=q/m$.
	Then $\gamma$ obeys the Lorentz force law:
	\begin{equation}\label{eqLorentz}
		\nabla_{\dot\gamma} \dot\gamma=-\rho\,\big(\iota_{\dot\gamma} \mathcal F\big)^\sharp,
	\end{equation}
	where $\iota_{\dot\gamma} \mathcal F:=\mathcal F(\dot\gamma,\cdot)$ and ${}^\sharp$ is the musical isomorphism ${}^\sharp:T^*M\to TM$ associated with $g$. 
	As $\mathcal F$ is antisymmetric, it is immediate to see that  arbitrary solutions $\gamma$ of \eqref{eqLorentz} have a well-defined causal character (see also \eqref{causalchar} below).
	Since both $g$ and $\mathcal A$ are stationary, the conservation law below is an immediate consequence of Noether's theorem. We nevertheless include a direct proof for completeness.
	\begin{proposition}
		Let $(M,g)$ be a stationary spacetime with a timelike Killing vector field $K$ and let $\mathcal A$ be a stationary electromagnetic potential, i.e. $\mathcal L_K\mathcal A=0$, with $\mathcal F=d\mathcal A$.
		Let $\gamma=\gamma(s)$   solve \eqref{eqLorentz}.
		Then the quantity 
		\begin{equation}\label{constmotion}
			\varepsilon:=-g(K,\dot\gamma)-\rho\,\mathcal A(K)
		\end{equation}
		is constant along $\gamma$.
		
		More generally, if $\mathcal L_K\mathcal A=d\psi$ for some smooth function $\psi$, then
		$\varepsilon+\rho\,\psi$ is constant along $\gamma$.
	\end{proposition}
	
	\begin{proof}
		Since $K$ is Killing, $g(\nabla_{\dot\gamma}K,\dot\gamma)=0$, hence differentiating $g(K,\dot\gamma)$ along $\gamma$, we get:
		\[
		\frac{d}{ds}g(K,\dot\gamma)=g(K,\nabla_{\dot\gamma}\dot\gamma).
		\]
		Using \eqref{eqLorentz} we obtain $\frac{d}{ds}g(K,\dot\gamma)=-\rho\,(\iota_{\dot\gamma}\mathcal F)(K)=-\rho\,\mathcal F(\dot\gamma,K)$.
		On the other hand, by Cartan's formula,
		\[
		\mathcal L_K\mathcal A=\iota_K d\mathcal A + d(\mathcal A(K))=\iota_K \mathcal F + d(\mathcal A(K)).
		\]
		If $\mathcal L_K\mathcal A=0$, then $d(\mathcal A(K))=-\iota_K \mathcal F$, hence
		\[
		\frac{d}{ds}\mathcal A(K)=d(\mathcal A(K))(\dot\gamma)=-(\iota_K\mathcal F)(\dot\gamma)=-\mathcal F(K,\dot\gamma)=\mathcal F(\dot\gamma,K).
		\]
		Therefore,
		\[
		\frac{d}{ds}\big(g(K,\dot\gamma)+\rho\,\mathcal A(K)\big)=
		-\rho \mathcal F(\dot\gamma,K)+\rho \mathcal F(\dot\gamma,K)=0.
		\]
		The  case $\mathcal L_K\mathcal A=d\psi$ is identical, since then
		$d(\mathcal A(K)-\psi)=-\iota_K\mathcal F$.
	\end{proof}
	The constant \eqref{constmotion} is called {\em energy} (per unit mass) of the solution $\gamma$ of \eqref{eqLorentz}.
	Its \emph{``kinetic part''} is
	\begin{equation*}
		\mathcal E_k \;:=\; -\,g(K,\dot\gamma) \;=\; \varepsilon \;+\; \rho\,\mathcal A(K)
	\end{equation*}
	but  in general it is not a constant of motion. 
	\begin{remark}[Stationary gauge transformations]\label{gauge-energy}
		Under a stationary gauge transformation
		\[
		\mathcal A\longmapsto\mathcal A+df,
		\qquad
		\mathcal L_K(df)=0,
		\]
		the function $K(f)$ is constant on $M$, say $K(f)=c$, and the energy
		transforms as $	\varepsilon\mapsto\varepsilon'=\varepsilon-\rho c$.
		On the other hand, the kinetic part of the energy is gauge invariant so that the numerical value of the conserved energy $\varepsilon$ depends on 	the stationary gauge, whereas $\mathcal E_k=-g(K,\dot\gamma)$ does not.
	\end{remark}
	\begin{remark}\label{futurepast}
		Let $p\in M$ and $u\in T_pM $ be a timelike unit vector. Let us consider the unique solution of \eqref{eqLorentz} satisfying the initial conditions $(p, u)$. 
		Then $\mathcal E_k>0$ iff $u$ is future-pointing (same time orientation as $K$),
		while $\mathcal E_k<0$ iff $u$ is past-pointing. The Lorentz force equation
		$\nabla_{\dot\gamma} \dot\gamma=-\rho(\iota_{\dot\gamma} \mathcal F)^\sharp$ is not invariant under orientation reversing reparametrization of $\gamma$
		at fixed $\rho$; it becomes invariant under simultaneous reversal of orientation and charge sign  ($q\mapsto -q$).
	\end{remark}

Let $\scrT\subset M$ be a smooth,   immersed,  timelike  hypersurface. 
\begin{definition}[Massive particle surface at fixed $(\rho,\varepsilon)$]\label{rhoepsMPS}
Fix a charge-to-mass ratio $\rho\in\mathbb R$ and an  energy $\varepsilon\in\mathbb R$.
The hypersurface $\scrT$ is a \emph{$(\rho,\varepsilon)$--massive 	particle surface} (briefly, a $(\rho,\varepsilon)$--MPS) if, for every 	$p\in\scrT$, there exists at least one vector $u\in T_p\scrT$ satisfying
$g(u,u)=-1$ and $-g(K,u)-\rho\mathcal A_p(K)=\varepsilon$,	and, for every such $u$, the solution of \eqref{eqLorentz} with initial 	conditions $\gamma(0)=p$ and $\dot\gamma(0)=u$ locally remains in $\scrT$.
\end{definition}
As for photon surfaces, 	if $\scrT$ is embedded and  closed as a subset of $M$,
then every maximal solution $\gamma:I\to M$ of 		\eqref{eqLorentz} with initial data as in	Definition~\ref{rhoepsMPS} remains in $\scrT$ throughout $I$.
\subsection{Extrinsic characterization of \texorpdfstring{$(\rho,\varepsilon)$}{(rho, epsilon)}--MPS}
	Let $g_\scrT$ be the metric induced by $g$ on $\scrT$ and let $\kappa$ be the orthogonal projection of the Killing field $K$ onto $T\scrT$; we set
	$\kappa^\flat:=g_\scrT (\kappa,\cdot)$, $\kappa^2:=-g_\scrT (\kappa,\kappa)$. Since $K$ is timelike and a unit normal to $\scrT$  is spacelike, $\kappa$ is timelike; hence $\kappa^2>0$, and $g_\scrT$ restricted to $\kappa^\perp\subset T\scrT$ is positive definite.
	\bd
	Fix $\varepsilon,\rho\in\R$ and let $u\in T_p\mathscr T$. We call the vector $u$ \emph{$(\rho,\varepsilon)$-admissible} if $g(u,u)=-1$ and $-g(u,K_p)=-g_\scrT(u,\kappa_p)=\varepsilon+\rho\A_p(K)$.
	\ed
	Assume  $\mathcal E_k:=\varepsilon +\rho\A(K)$ is nowhere vanishing on $\scrT$, and define the symmetric tensors on $T\mathscr T$
	\beq\label{HF}
	\H :=g_\scrT +\ \mathcal E_k^{-2}\,\kappa^\flat\!\otimes\kappa^\flat,\qquad
	\mathcal I := \frac12\Big((\iota_n \mathcal F)\otimes\kappa^\flat\ +\ \kappa^\flat\otimes(\iota_n \mathcal F)\Big).
	\eeq
	Let $D_p:=\{X\in T_p\mathscr T:\ g(X,\kappa)=0\}$ be the $(d-2)$-plane orthogonal to $\kappa$ (inside $T_p\mathscr T$).
	The \emph{unnormalized $\kappa$-orthogonal mean curvature} is $H_\kappa(p)\ :=\ \sum_{a=1}^{d-2}\Pi(e_a,e_a)$, 	for any $g_{\scrT}$-orthonormal basis $\{e_a\}$ of $D_p$. Notice that, although $\Pi$, $H_\kappa$, and $\mathcal I$ depend on the choice of the unit normal $n$, the identities below do not. Indeed, under the replacement $n\mapsto -n$,  $\Pi$, $H_\kappa$, and $\mathcal I$ change sign simultaneously.
	
	The following result is essentially \cite[Thm.~3.1]{KBG2022}, where it is stated
	for an arbitrary Killing vector field whose projection $\kappa$ on $\scrT$ is
	timelike, spacelike or null (see \cite[App.~A]{KBG2022} for the null case).
	We restrict to a timelike Killing field and reformulate it in coordinate-free
	form for solutions of \eqref{eqLorentz} parametrized by proper time with fixed
	charge-to-mass ratio. Moreover, we prove the implication (2)$\Rightarrow$(1) by
	showing directly that the vector field generating the Lorentz flow is tangent to
	the set of admissible vectors. In the corresponding step of the proof of
	\cite[Thm.~3.1]{KBG2022}, the constraint $-g(\kappa,\dot\gamma)=\mathcal E_k$ is
	used along the whole worldline of the dynamics induced on $\scrT$. This is
	immediate when $K$ is tangent to $\scrT$ but requires a further argument  in general. 
	\begin{theorem}\label{thmumbMPS}
		Let $\varepsilon,\rho\in\R$ and set $\mathcal E_k:=\varepsilon+\rho\,\mathcal A(K)$.
		Let us assume $\mathcal E_k^2>\kappa^2$ on $\scrT$ (hence $\mathcal E_k\neq 0$ on $\scrT$).
		Then the following are equivalent:	
		\begin{enumerate}\itemsep=2pt
			\item $\mathscr T$ is a $(\rho,\varepsilon)$--MPS;
			\item for every $(\rho,\varepsilon)$-admissible $u\in T\mathscr T$, 
			\beq\label{Piu}
			\Pi(u,u) =\rho\,\mathcal F(u,n);\eeq
			\item
			\beq\label{Pi}
			\Pi\ =\ \frac{H_\kappa}{d-2}\H + \frac{\rho}{\mathcal E_k}\,\mathcal I.
			\eeq
		\end{enumerate}
	\end{theorem}
	\begin{proof}
		\noindent\emph{$(1)\Rightarrow(2)$.}
		Let $u\in T\scrT$ be a $(\rho,\varepsilon)$-admissible vector and take the solution $\gamma=\gamma(\tau)$ of  \eqref{eqLorentz} with initial conditions defined by $u$. 
		By decomposing the ambient covariant derivative $\nabla_{\!u}u=\nabla^{\scrT}_{\!u}u-\Pi(u,u)\,n$,
		and projecting onto $n$,  using $g(\mathcal F(u,\cdot)^{\sharp},n)=\mathcal F(u,n)$, yields
		\[
		\Pi(u,u)=\rho\,\mathcal F(u,n).
		\]
		\smallskip\noindent\emph{$(2)\Rightarrow(3)$.}
		Fix $p\in\scrT$ and write any $(\rho,\varepsilon)$-admissible $u\in T_p\scrT$ as
		$u=\alpha\,\kappa+v$ with $v\perp\kappa$ (with respect to $g_{\scrT}$).
		From  $g_{\scrT}(u,\kappa)=-\mathcal E_k$ and $g_{\scrT}(u,u)=-1$ we get
		\[
		\alpha=\frac{\mathcal E_k}{\kappa^2},\qquad		g_{\scrT}(v,v)=\alpha^2\kappa^2-1=:r^2>0,
		\]
		hence $v$ ranges over the sphere $S_r:=\{\,v\in\kappa^\perp:\ g_{\scrT}(v,v)=r^2\,\}$.
		Set, for $v\in\kappa^\perp$,
		\[
		P(v):=\Pi(u,u)-\rho \mathcal F(u, n),\qquad u=\alpha\kappa+v.
		\]
		Assumption (2) is equivalent to $P(v)=0$ for all $v\in S_r$. Expanding by bilinearity,
		\[
		P(v)=\Pi(v,v)+2\alpha\,\Pi(\kappa,v)+\alpha^2\Pi(\kappa,\kappa)
		+\rho\alpha\,\iota_n \mathcal F(\kappa)+\rho\,\iota_n \mathcal F(v).
		\]
		Replace $v$ by $-v$ and subtract:
		\[
		0=P(v)-P(-v)=4\alpha\,\Pi(\kappa,v)+2\rho\,\iota_n \mathcal F(v)\qquad (v\in S_r).
		\]
		The right-hand side is linear in $v$ and vanishes on the sphere $S_r$,
		hence it vanishes on all of $\kappa^\perp$; thus
		\beq\label{Pikappa}
		\Pi(\kappa,\cdot)=-\frac{\rho}{2\alpha}\,\iota_n \mathcal F\quad\text{on }\kappa^\perp.
		\eeq
		On the other hand,
		\beq\label{Fkappa}
		\mathcal I(\kappa,v)=\frac12\big(\iota_n \mathcal F(\kappa)\kappa^\flat(v)+\kappa^\flat(\kappa)\iota_n \mathcal F(v)\big)
		=\frac12\,g_{\scrT}(\kappa,\kappa)\,\iota_n \mathcal F(v)=-\frac{\kappa^2}{2}\,\iota_n \mathcal F(v).
		\eeq
		Since $\alpha=\mathcal E_k/\kappa^2$,  \eqref{Pikappa} and \eqref{Fkappa} give:
		\[
		\Pi(\kappa,\cdot)=\frac{\rho}{\mathcal E_k}\mathcal I(\kappa,\cdot)\quad\text{on }\kappa^\perp.
		\]
		We also have 
		\beq\label{PvPmenov}
		0=P(v)+P(-v)=2\,\Pi(v,v)+2\alpha^2\Pi(\kappa,\kappa)+2\rho\alpha\,\iota_n \mathcal F(\kappa)\qquad(v\in S_r).
		\eeq
		Hence the quadratic form $v\mapsto\Pi(v,v)$ is constant on the sphere $S_r$.
		Therefore, on the inner product space $(\kappa^\perp,g_{\scrT}|_{\kappa^\perp})$,
		\beq\label{Piperp}
		\Pi|_{\kappa^\perp\times\kappa^\perp}=\lambda\,g_{\scrT}|_{\kappa^\perp\times\kappa^\perp},
		\qquad \lambda=\frac{H_\kappa}{d-2}
		\eeq
		(by taking the $g_{\scrT}$-trace on $\kappa^\perp$). We still have to determine  $\Pi(\kappa,\kappa)$.
		From \eqref{PvPmenov} and \eqref{Piperp}, recalling that $g_{\scrT}(v,v)=r^2=\alpha^2\kappa^2-1$, we get
		\[
		\alpha^2\,\Pi(\kappa,\kappa)
		=-\,\lambda\,r^2\;-\;\rho\alpha\,\iota_n \mathcal F(\kappa)
		=-\,\frac{H_\kappa}{d-2}\,(\alpha^2\kappa^2-1)\;-\;\rho\alpha\,\iota_n \mathcal F(\kappa).
		\]
		Dividing by $\alpha^2$ and using $\alpha=\mathcal E_k/\kappa^2$,
		\[
		\Pi(\kappa,\kappa)
		=-\,\frac{H_\kappa}{d-2}\Big(\kappa^2-\frac{\kappa^4}{\mathcal E_k^2}\Big)
		\;-\;\frac{\rho}{\mathcal E_k}\,\kappa^2\,\iota_n \mathcal F(\kappa).
		\]
		From \eqref{HF}, we get:
		\[
		\mathcal H(\kappa,\kappa)=-\kappa^2+\frac{\kappa^4}{\mathcal E_k^2},\qquad
		\mathcal I(\kappa,\kappa)=-\,\kappa^2\,\iota_n \mathcal F(\kappa).
		\]
		Therefore
		\[
		\Pi(\kappa,\kappa)
		=\frac{H_\kappa}{d-2}\,\mathcal H(\kappa,\kappa)
		+\frac{\rho}{\mathcal E_k}\,\mathcal I(\kappa,\kappa).
		\]
		Together with
		$\Pi|_{\kappa^\perp\times\kappa^\perp}=\frac{H_\kappa}{d-2}\,g_{\scrT}|_{\kappa^\perp\times\kappa^\perp}$
		and $\Pi(\kappa,\cdot)=\frac{\rho}{\mathcal E_k}\,\mathcal I(\kappa,\cdot)$ on $\kappa^\perp$,
		and since $\mathcal I|_{\kappa^\perp\times\kappa^\perp}=0$ and
		$\mathcal H|_{\kappa^\perp\times\kappa^\perp}=g_{\scrT}|_{\kappa^\perp\times\kappa^\perp}$,
		we conclude
		that  \eqref{Pi} holds.
		
		\smallskip\noindent\emph{$(3)\Rightarrow(2)$.}
		Let $u$ be $(\rho,\varepsilon)$-admissible. Then by definition of $\H$,
		\[
		\H(u,u)=g_\scrT(u,u)+\mathcal E_k^{-2}\,\kappa^\flat(u)^2=-1+\mathcal E_k^{-2}\mathcal E_k^2=0.
		\]
		Therefore, evaluating \eqref{Pi} on $(u,u)$ and using
		\[
		\mathcal I(u,u)=\iota_n \mathcal F(u)\,\kappa^\flat(u)=\iota_n \mathcal F(u)\,g(u,\kappa)=-\mathcal E_k\,\iota_n \mathcal F(u),
		\]
		we obtain $\Pi(u,u)=-(\rho/\mathcal E_k)\,\mathcal E_k\,\iota_n \mathcal F(u)=-\rho\,\iota_n \mathcal F(u)$, i.e. (2).
		
		\medskip\noindent\emph{$(2)\Rightarrow(1)$.}
		Let $\pi:TM\to M$ be the canonical projection. The Lorentz force 	dynamics is generated by the vector field $\mathcal X$ on $TM$ which,
		in local coordinates, is given by
		\[
		\mathcal X(x,u)=u^i\partial_{x^i}+\left(-\Gamma^i_{jk}(x)u^ju^k	
			+\rho\,g^{i\ell}(x)\mathcal F_{\ell m}(x)u^m\right)\partial_{u^i}.
		\]
		Thus, along an integral curve $(\sigma(\tau),u(\tau))$ of $\mathcal X$,
		one has $\dot \sigma=u$ and $\nabla_u u=-\rho\,(\iota_u\mathcal F)^\sharp$.
		We first consider the set of $(\rho,\varepsilon)$-admissible vectors tangent to $\scrT$, i.e.
		\[
		\mathcal W_{\rho,\varepsilon}:=\left\{(x,u)\in T\mathscr T:	g(u,u)=-1,\quad
		-g(K,u)-\rho\mathcal A(K)=\varepsilon\right\}.
		\]
		Let
		\[
		\mathcal C:T\mathscr T\longrightarrow\mathbb R^2,\qquad	\mathcal C(x,u)	=
		\bigl(g(u,u)+1,\,-g(K,u)-\rho\mathcal A(K)-\varepsilon\bigr).
		\]
		Then $\mathcal W_{\rho,\varepsilon}=\mathcal C^{-1}(0,0)$. At $(x,u)\in\mathcal W_{\rho,\varepsilon}$, restricting 	$d\mathcal C$ to the vertical directions
		$w\in T_x\mathscr T$, we have
		\[
		d\mathcal C_1(0,w)=2g(u,w),	\qquad	d\mathcal C_2(0,w)=-g(\kappa,w).
		\]
		Since $u$ is $(\rho,\varepsilon)$-admissible,
		\[
		u=\frac{\mathcal E_k}{\kappa^2}\kappa+v,
		\qquad
		v\perp\kappa,
		\]
		and
		\[
		g(v,v)=\frac{\mathcal E_k^2}{\kappa^2}-1>0.
		\]
		Hence $u$ and $\kappa$ are not collinear. Since 	$g_{\mathscr T}$ is non-degenerate, the covectors 	$g(u,\cdot)$ and $g(\kappa,\cdot)$ on $T_x\mathscr T$
		are linearly independent. Therefore $d\mathcal C$ has rank $2$
		on $\mathcal W_{\rho,\varepsilon}$, and the regular value theorem
		shows that $\mathcal W_{\rho,\varepsilon}$ is a smooth submanifold
		of $T\mathscr T$ of codimension $2$.

		We now work locally near a point $p\in\mathscr T$.
		Choose a relatively compact neighborhood $V\subset\mathscr T$ of
		$p$ and a tubular neighborhood $\mathcal U$ of $V$ such that the
		normal exponential map
		\[
		\mathfrak E:(-\delta,\delta)\times V\longrightarrow\mathcal U,
		\qquad
		\mathfrak E(s,q)=\exp_q(s\,n_q),
		\]
		is a diffeomorphism. Define
		$\phi:\mathcal U\to\mathbb R$ by $\phi(\mathfrak E(s,q))=s$.	Then, along $V$,
	$\phi=0$ and $\nabla\phi=n$,
		and hence, for $X,Y\in TV$,
		\[
		\operatorname{Hess}\phi(X,Y)=g(\nabla_Xn,Y)	=\Pi(X,Y).
		\]
		On $T\mathcal U$ define
		\[
		\Phi:=\phi\circ\pi,	\qquad	\Psi(x,u):=g_x(\nabla\phi(x),u).
		\]
		Then
		\[
		TV=	\{(x,u)\in T\mathcal U:\Phi(x,u)=0,\ \Psi(x,u)=0\}.
		\]
		We now compute
\[
\Lie_{\mathcal X}\Phi= d(\phi\circ\pi)(\mathcal X)= d\phi_x\!\big(d\pi\,\mathcal X\big)= d\phi_x(u)=g(\nabla \phi,u),
\]
	and, using metric compatibility,
		\[
		\Lie_{\mathcal X}\Psi=g(\nabla_u\nabla\phi,u)+g(\nabla\phi,\nabla_u u)=
		\operatorname{Hess}\phi_x(u,u)-\rho \mathcal F(u,\nabla\phi).
		\]
		Therefore, on $\mathcal W_{\rho,\varepsilon}\cap TV$, we have  
		$\Lie_{\mathcal X}\Phi=0$ and, by~\eqref{Piu}, $\Lie_{\mathcal X}\Psi=	\Pi(u,u)-\rho \mathcal F(u,n)=0$.
		Moreover, since $g(u,u)$ and $-g(K,u)-\rho\mathcal A(K)$ are first integrals of the
		Lorentz force dynamics, their Lie derivatives along $\mathcal X$ 	vanish as well. Hence $\mathcal X$ is tangent to 	$\mathcal W_{\rho,\varepsilon}$.

		It follows  that an
		integral curve of $\mathcal X$ starting from a point  in $\mathcal W_{\rho,\varepsilon}$ remains in 	$\mathcal W_{\rho,\varepsilon}$.  Thus $\mathscr T$ is a $(\rho,\varepsilon)$--MPS.
		\end{proof}
	
	\begin{remark}[The threshold energy level $\mathcal E_k^2=\kappa^2$]\label{degenerate}
		In the proof of $(2)\Rightarrow(3)$ we used that the $(\rho,\varepsilon)$-admissible vectors at $p\in\scrT$
		form a non-degenerate sphere $S_r\subset\kappa^\perp$, where
		\[
		r^2=\frac{\mathcal E_k^2}{\kappa^2}-1>0.
		\]
		Recall that $\mathcal E_k=\varepsilon +\rho\mathcal A(K)$. 		If $\mathcal E_k^2=\kappa^2$, then $r=0$ and  the only $(\rho,\varepsilon)$-admissible vector   is  		$u=\pm \kappa/\sqrt{\kappa^2}$,
		with the sign determined by $\mathcal E_k$.
		In this degenerate situation,  \eqref{Piu}  does not determine $\Pi$ on
		$\kappa^\perp$. Therefore the implication $(2)\Rightarrow(3)$ cannot be derived
		from \eqref{Piu} when $\mathcal E_k^2=\kappa^2$ (while $(3)\Rightarrow(2)$ still holds).		
		If $\mathcal E_k^2<\kappa^2$, there are no $(\rho,\varepsilon)$-admissible vectors  at $p$.

		Thus a $(\rho,\varepsilon)$--MPS necessarily satisfies 		$\mathcal E_k^2\geq\kappa^2$.
		The equality case is degenerate as explained above, hence all the extrinsic consequences below
		that use \eqref{Pi} are established under the condition $\mathcal E_k^2>\kappa^2$.
		This is precisely the hypothesis $\kappa^2>-\mathcal E_k^2/m^2$ of
		\cite[Thm.~3.1]{KBG2022}; there $\kappa^2$ denotes $g(\kappa,\kappa)$,
		which is negative in our setting, and the tangent vector of a
		worldline is normalized by $g(\dot\gamma,\dot\gamma)=-m^2$.  Hence, with our convention $\kappa^2=-g(\kappa,\kappa)$, 	the two inequalities coincide.
		\end{remark}

\subsection{Quasi-umbilicity of massive particle hypersurfaces}\label{qumb}
We now consider Chen's notion of \emph{quasi-umbilicity}
\cite[Ch.~5]{Chen1973}, originally formulated for submanifolds of
Riemannian manifolds. We use below its natural  analogue for non-degenerate  hypersurfaces.

Recall that the shape operator of a non-degenerate hypersurface
$\scrT$ is $S_p\colon T_p\scrT\longrightarrow T_p\scrT$, $S_p(v):=-\nabla_v n$,
so that, with our convention for the second fundamental form, $\Pi_p(v,w)=-g_{\scrT}(S_p(v),w)$.
Let 
\[
U:=\bigl\{p\in\scrT:\Pi_p\text{ is not proportional to }
g_{\scrT,p}\bigr\}
\]
be the non-umbilical locus. In the Riemannian setting this is precisely
the open set used by Chen, namely the set where the shape operator has
at least two distinct eigenvalues.

Following Chen's tensorial formulation, we say that $\scrT$ is \emph{quasi-umbilical} if  locally on $U$ there exist two smooth functions $a,b$ and a smooth unit one-form $\vartheta$ such that 
\[g_{\scrT}(\vartheta^\sharp,\vartheta^\sharp)\in\{-1,1\} \qquad \text{and}\qquad \Pi=a\,g_{\scrT}+b\,\vartheta\otimes\vartheta.\]
Notice that  $b$ is necessarily nowhere vanishing, that   the distinguished direction $\vartheta^\sharp$ is required to be non-null, and that on $\scrT\setminus U$ the hypersurface is  umbilical. 

In our  setting, let $D:=\kappa^\perp\subset T\scrT$. Since $\kappa$ is timelike, $D$ is spacelike and has rank $d-2$.
For every $(\rho,\varepsilon)$--MPS satisfying
$\mathcal E_k^2>\kappa^2$, formula~\eqref{Pi} gives
\begin{equation}\label{PD}
\Pi|_{D\times D}=\frac{H_\kappa}{d-2}\,g_{\scrT}|_{D\times D}.
\end{equation}
The proof of the following proposition shows that quasi-umbilicity is equivalent
to the vanishing of the mixed components of $\Pi$ with respect to the
orthogonal decomposition $T\scrT=D\oplus\operatorname{span}\{\kappa\}$,
and characterizes this condition in terms of the electromagnetic field or the vanishing of the electric charge.
\begin{proposition}\label{quasiumbilicequiv}
Assume $d\geq4$, $\mathcal E_k^2>\kappa^2$ on $\scrT$. Let $\scrT$ be a $(\rho,\varepsilon)$--MPS. Then the following conditions are	equivalent:
	\begin{enumerate}
		\item $\scrT$ is quasi-umbilical;
		\item either $\rho=0$ or $\iota_n\mathcal F=c\,\kappa^\flat$
		for some smooth function $c$ on each domain where $n$ is defined, possibly vanishing.
	\end{enumerate}
\end{proposition}
\begin{proof}
Let $\scrT$ be quasi-umbilical and $p\in U$. By \eqref{PD} we have
\[
\left(\frac{H_\kappa}{d-2}-a\right)	g_{\scrT}|_{D\times D}	=	b\,\vartheta|_D\otimes\vartheta|_D.
\]
We notice that since  $\dim D=d-2\geq2$ and $g_{\scrT}|_D$ is positive definite, if
\[
\frac{H_\kappa}{d-2}-a\neq0,
\]
the left-hand side has rank $d-2\geq2$, whereas the right-hand side has rank at most one. Therefore, necessarily 
\[
a=\frac{H_\kappa}{d-2},	\qquad	b\,\vartheta|_D\otimes\vartheta|_D=0.
	\]
Since $p\in U$,  $b(p)\neq0$, and the identity $b(p)\,\vartheta_p|_{D_p}\otimes\vartheta_p|_{D_p}=0$
implies $\vartheta_p|_{D_p}=0$. Thus, $D_p=\ker\vartheta_p$, since both spaces have dimension $d-2$.
If instead $p\notin U$, then $\Pi_p$ is proportional to 	$(g_{\scrT})_p$, and the mixed components $\Pi_p(\kappa_p,v)$ vanish as well. Thus, for each $p\in\scrT$ and $v\in D_p$ we have $\Pi_p(\kappa_p,v)=0$.
Using~\eqref{Pi}, we obtain
\[	0	=	\Pi_p(\kappa_p,v)	=	-\frac{\rho\kappa^2(p)}{2\mathcal E_k(p)}(\iota_n\mathcal F)_p(v).\]
Since $\mathcal E_k^2>\kappa^2$, we have $\mathcal E_k\neq0$. 	Consequently, if $\rho\neq0$, for all 
$v\in D_p$ we must have $(\iota_n\mathcal F)_p(v)=0$ and then we deduce that
\[
\iota_n\mathcal F=c\,\kappa^\flat,	\qquad	c=-\frac{\iota_n\mathcal F(\kappa)}{\kappa^2}.
\]

For the converse implication, if $\rho=0$, formula~\eqref{Pi} gives
\[	\Pi	=	\frac{H_\kappa}{d-2}\,g_{\scrT}	+	\frac{H_\kappa}{d-2}\mathcal E_k^{-2}
	\kappa^\flat\otimes\kappa^\flat.
\]
If instead $\iota_n\mathcal F=c\,\kappa^\flat$,	then, by~\eqref{HF}, $\mathcal I	=	c\,\kappa^\flat\otimes\kappa^\flat$,
and therefore
\[
\Pi	=	\frac{H_\kappa}{d-2}\,g_{\scrT}	+	\left(	\frac{H_\kappa}{d-2}\mathcal E_k^{-2}	+
	\frac{\rho}{\mathcal E_k}c	\right)	\kappa^\flat\otimes\kappa^\flat.
\]
In either case, wherever the rank-one coefficient does not vanish,	we may take
$\vartheta:=\kappa^\flat/\sqrt{\kappa^2}$, 
and $\Pi$ has Chen's quasi-umbilical form. Wherever that coefficient 	vanishes, $\Pi$ is proportional to $g_{\scrT}$ and the hypersurface is 	umbilical. Hence $\scrT$ is quasi-umbilical.
\end{proof}
\begin{remark}
Under the hypotheses of Proposition~\ref{quasiumbilicequiv}, \eqref{PD} concerns  only 	the restriction of $\Pi$ to $D\times D$, where 	$D=\kappa^\perp$. In fact, 
\[
S(v)=-\frac{H_\kappa}{d-2}v +\frac{\Pi(v,\kappa)}{\kappa^2}\kappa, 	\qquad v\in D.
\]
Thus the shape operator need 	not preserve $D$ and may fail to be diagonalizable.
As shown in the proof of Proposition~\ref{quasiumbilicequiv},  	quasi-umbilicity of a $(\rho,\varepsilon)$--MPS is equivalent to the vanishing of the
mixed components $\Pi(\kappa,v)$ for all $v\in D$. 	In that case, the orthogonal splitting 	$T\scrT=D\oplus\operatorname{span}\{\kappa\}$ 	diagonalizes $S$. On the non-umbilical locus $U$, 	$D$ is therefore a spacelike principal distribution, i.e. 	its fibers are eigenspaces of $S$ on which the induced 	metric is positive definite. The corresponding principal 	curvature is $-H_\kappa/(d-2)$, with multiplicity $d-2$. 	
This is analogous to  the $(d-2)$-umbilic condition of \cite{DajczerFlorit1998}, formulated for Riemannian submanifolds of space forms.  In the hypersurface case, the definition in \cite{DajczerFlorit1998} requires the corresponding principal curvature to be nonzero, which here amounts to $H_\kappa\neq0$, 
while  the quasi-umbilical condition permits $H_\kappa=0$, 	even on $U$ in the charged case. For an earlier treatment of the Euclidean hypersurface case in terms of envelopes of spheres, see
\cite{AspertiDajczer1984} and for related work on  $(d-2)$-umbilic spacelike hypersurfaces in a Lorentzian manifold of dimension $d$, see \cite{ColaresPalmas2008}.
\end{remark}

\begin{remark}[The geodesic case]
	If  $\rho=0$ or $\mathcal F=0$ (taking then also $\mathcal A=0$), the Lorentz
	force equation reduces to the timelike geodesic equation and
	$\mathcal E_k=\varepsilon$. Hence, if
	$\varepsilon^2>\kappa^2$, Theorem~\ref{thmumbMPS} gives
	\[
	\Pi	=	\frac{H_\kappa}{d-2}	\Bigl(	g_{\scrT}	+	\varepsilon^{-2}\kappa^\flat\otimes\kappa^\flat	\Bigr).
	\]
	In particular, for $d\geq4$ and $\varepsilon^2>\kappa^2$, any $(0,\varepsilon)$--MPS is quasi-umbilical. The energy
	label $\varepsilon$ remains meaningful because the proper-time
	normalization $g(\dot\gamma,\dot\gamma)=-1$ is fixed, consistently
	with the fixed-mass, fixed-energy viewpoint adopted in
	\cite{BerCardAnd25}.
\end{remark}

\subsection{MPS master equation}
We recover, in our conventions, the master equation of \cite[Eq.~(35)]{KBG2022}, which gives a scalar equation
that any $(\rho,\varepsilon)$--MPS must satisfy. 
Set
\[
\mathcal G:=\frac{\iota_n\mathcal F(\kappa)}{\kappa^2}=\frac{\mathcal F(n,\kappa)}{\kappa^2},\qquad
\mathcal K:=H_\kappa+(d-2)\frac{\Pi(\kappa,\kappa)}{\kappa^2}.
\]
\begin{proposition}\label{master}
	Assume that $\mathcal E_k^2>\kappa^2$ on $\scrT$ and that
	$\scrT$ is a $(\rho,\varepsilon)$--MPS. Then
	\begin{equation}\label{eq:master-quadratic}
		\mathcal K\,\mathcal E_k^2
		+(d-2)\rho\kappa^2\mathcal G\,\mathcal E_k
		-\kappa^2H_\kappa=0,
	\end{equation}
	and the discriminant $\Delta:=(d-2)^2\rho^2\kappa^4\mathcal G^2	+4\mathcal K\kappa^2H_\kappa$
	is nonnegative on $\scrT$.
	At each point where $\mathcal K\neq0$, the kinetic energy
	$\mathcal E_k$ equals one of the real roots
	\begin{equation*}
		\mathcal E_k^\pm=	\frac{-(d-2)\rho\kappa^2\mathcal G	\pm\sqrt{\Delta}}{2\mathcal K}.
	\end{equation*}
	At points where $\mathcal K=0$,
	\eqref{eq:master-quadratic} reduces to $(d-2)\rho\mathcal G\,\mathcal E_k=H_\kappa$.
	Thus, if $\rho\mathcal G\neq0$, then
	$\mathcal E_k	=H_\kappa/\big((d-2)\rho\mathcal G\big)$;	
	if $\rho\mathcal G=0$, the equation reduces to $H_\kappa=0$.
\end{proposition}

\begin{proof}
	Evaluating \eqref{Pi} on $(\kappa,\kappa)$ and using
	\eqref{HF} gives
	\[
	\Pi(\kappa,\kappa)
	=
	-\frac{H_\kappa}{d-2}
	\left(\kappa^2-\frac{\kappa^4}{\mathcal E_k^2}\right)
	-\frac{\rho}{\mathcal E_k}\,
	\kappa^2\,\iota_n\mathcal F(\kappa).
	\]
	Multiplying by $(d-2)/\kappa^2$ and using
	$\iota_n\mathcal F(\kappa)=\kappa^2\mathcal G$, we obtain
	\[
	(d-2)\frac{\Pi(\kappa,\kappa)}{\kappa^2}
	=	-H_\kappa	+\frac{H_\kappa\kappa^2}{\mathcal E_k^2}
		-(d-2)\frac{\rho\kappa^2\mathcal G}{\mathcal E_k}.
	\]
	Adding $H_\kappa$ and recalling the definition of
	$\mathcal K$ yields
	\[
	\mathcal K	=	\frac{H_\kappa\kappa^2}{\mathcal E_k^2}	
			-(d-2)\frac{\rho\kappa^2\mathcal G}{\mathcal E_k}.
	\]
	Since $\mathcal E_k\neq0$, this is equivalent to
	\eqref{eq:master-quadratic}.
	That equation also gives $\Delta=\bigl(	2\mathcal K\mathcal E_k	+(d-2)\rho\kappa^2\mathcal G
	\bigr)^2\geq0$.	The remaining assertions follow by solving	\eqref{eq:master-quadratic}.
\end{proof}
\begin{remark}
The nonnegativity of $\Delta$ in Proposition~\ref{master} 	is a consequence of the assumption that $\scrT$ is already 	a $(\rho,\varepsilon)$--MPS.	For an arbitrary timelike hypersurface 
being tested for 	this property, $\Delta$ need not be nonnegative.
If $\Delta<0$ at any point, the master equation has no 	real solution there, which rules out the MPS property for the prescribed charge-to-mass ratio.
	
Even when $\Delta\geq0$, the algebraic roots must satisfy 	the admissibility conditions. For future-pointing solutions, $\mathcal E_k=-g(K,\dot\gamma)>0$; together
with the non-degeneracy assumption, this requires $\mathcal E_k>\sqrt{\kappa^2}$.
Only roots satisfying this inequality are eligible. 
Finally, the existence of an admissible real root is also only a necessary condition. A candidate kinetic energy 	must satisfy $\mathcal E_k-\rho\mathcal A(K)\equiv\varepsilon$ on $\scrT$
for the prescribed constant $\varepsilon$, and the full 	tensorial identity \eqref{Pi} must hold.  	The master equation records only  the 	$(\kappa,\kappa)$ value  of $\Pi$.
\end{remark}
\begin{remark}[Comparison with the conventions of	\cite{KBG2022}]
	To compare \eqref{eq:master-quadratic} with
	\cite[Eq.~(35)]{KBG2022}, recall that the signed quantity
	denoted there by $\kappa^2$ is $g(\kappa,\kappa)<0$,
	whereas here $\kappa^2=-g(\kappa,\kappa)>0$.
	Moreover, if $\mathcal E_k^{\rm KBG}$ and $q$ denote
	the kinetic energy and charge used in \cite{KBG2022},
	proper-time normalization gives $\mathcal E_k^{\rm KBG}=m\mathcal E_k$ and 
	$q=m\rho$.
	Their $\chi_\tau$ and $\mathfrak F$ correspond,
	respectively, to $H_\kappa$ and $\iota_n\mathcal F(\kappa)$.
	With these identifications, their quantity $K$ in
	\cite[Eq.~(37)]{KBG2022} equals $-\mathcal K$, and their
	master equation becomes \eqref{eq:master-quadratic}.
\end{remark}

\section{Reduction of the Lorentz force equation to a non-relativistic electromagnetic system in a stationary spacetime}\label{KKJM}
Throughout this section we consider the Lorentz force dynamics as the Euler--Lagrange
	equation of the functional
	\begin{equation}\label{eq:functional_f}
		\mathscr F(\gamma)=\int_{a}^{b}\Big(\tfrac12\,g(\dot\gamma,\dot\gamma)+\rho\,\mathcal A(\dot\gamma)\Big)\,ds,
		\qquad \mathcal F=d\mathcal A,
	\end{equation}
	defined on $H^1$ curves, parametrized on a given interval $[a,b]$, with fixed endpoints.
	The parameter $s$  is not arbitrary (the functional is not reparametrization invariant) and its critical points satisfy \eqref{eqLorentz} with respect to $s$.
	Along any critical point $\gamma$ of \eqref{eq:functional_f} the quantity
	\begin{equation}\label{defe}
		e:=\tfrac12\,g(\dot\gamma,\dot\gamma)
	\end{equation}
	is constant. Indeed,
	\begin{equation}\label{causalchar}
		\frac{d}{ds}g(\dot\gamma,\dot\gamma)=2\,g(\nabla_{\dot\gamma}\dot\gamma,\dot\gamma)
		=-2\rho\,\mathcal F(\dot\gamma,\dot\gamma)=0.
	\end{equation}
	If $\gamma$ is timelike ($e<0$), then  $s$ coincides with proper time if and only if
	\begin{equation*}
		g(\dot\gamma,\dot\gamma)=-1\quad\Longleftrightarrow\quad e=-\tfrac12.
	\end{equation*}
The strongly indefinite functional \eqref{eq:functional_f} has been extensively studied from a variational viewpoint; see, in particular, \cite{CaMa02} for standard stationary spacetimes, \cite{CaMaPi04} for general stationary spacetimes, and \cite{CapCor2023,CapCorGiaPic24} for the more general setting of indefinite Lagrangians admitting an infinitesimal symmetry with affine Noether charge. In particular, \cite{CapCorGiaPic24} addresses the existence of solutions with prescribed Lagrangian energy, which for the Lagrangian in \eqref{eq:functional_f} is precisely the constant $e$ in \eqref{defe}. Our purpose here is different: rather than developing the variational theory of \eqref{eq:functional_f}, we work directly with the equations of motion and their reduction to a mechanical system on the spacelike slice.

Let us consider a standard stationary spacetime so that $M=\R\times S$ and  $g$ as in \eqref{1a}. Let us assume that the electromagnetic field $\mathcal F$ admits a stationary electromagnetic potential one-form $\mathcal A$  (so that $\mathcal L_K \mathcal A=0$ and $\mathcal F=d\mathcal A$). The  potential $\mathcal{A}$ does not depend on $t$ and it can be identified with $\mathcal A^\sharp = (\phi, A)$, where $\phi$ and $A$ are a  function  and   a vector field on $S$ respectively, such that, for any $(\tau, v)$
	\begin{equation*} 
		\mathcal{A} ( (\tau, v) ) =  
		- \beta \left(\phi -\hat{\omega}  (A) \right) \left(\tau  -\hat{\omega}  (v) \right)  +h_0 (A,v),
	\end{equation*}
	where $\hat \omega$ and $h_0$ are defined as in \eqref{opticaldata}.

For a causal critical curve  $\gamma = (t,\sigma)$,  the energy in \eqref{constmotion} becomes  
\[\varepsilon  = -\,g(\partial_t,\dot\gamma )\;-\rho \mathcal A(\partial_t) 
		 =  \beta  \left(\dot t -\hat{\omega}  (\dot \sigma) \right) + \rho \beta   \left(\phi -\hat{\omega}  (A) \right) 
		= \beta \left(\dot t -\hat{\omega }  (\dot \sigma+ \rho A ) + \rho \phi \right).\]
Let us introduce the function 
	\begin{equation}\label{chi}
		\chi:S\to \R,\quad \chi=\phi-\hat\omega(A),
	\end{equation}	
	so that 
	\[\varepsilon=\beta\big(\dot t-\hat\omega(\dot \sigma)+\rho\,\chi(\sigma)\big).\]
	
\begin{remark}
For later use, we note the following. We can write the stationary potential as $\mathcal A=\mathcal A_t dt+\alpha$,
where $\mathcal A_t:=\mathcal A(\partial_t)$ and  $\alpha$ is the pullback of $\mathcal A$ to a slice
of constant $t$, identified with a one-form on $S$.
Comparing $\mathcal A=\mathcal A_t dt+\alpha$ with 
\[
\mathcal A
=-\beta\chi(dt-\widehat\omega)+A^{\flat,h_0}
\]
gives
\begin{equation}\label{potential-components}
	A=(\alpha+\mathcal A_t\widehat\omega)^{\sharp,h_0},
	\qquad 	\chi=-\frac{\mathcal A_t}{\beta}, \qquad \phi=\widehat\omega(A)-\frac{\mathcal A_t}{\beta}.
\end{equation}
In particular, 
\begin{equation}\label{AS}
\mathcal A_S:=A^{\flat,h_0}
=\alpha+\mathcal A_t\widehat\omega.
\end{equation}
\end{remark}	
We now state a generalization of a  result, proved  in \cite{ger}  for geodesics, to solutions of \eqref{eqLorentz} parametrized with respect to $s$.
More precisely, the projection onto $S$ of such solutions satisfies the equation of a particle moving on $(S, h_0)$ under the action of a potential $V$ and a magnetic field $B$, which, for a fixed pair $(\rho, \varepsilon)$,  are given by
\begin{align}
		V&:=-\frac{\big(\varepsilon-\rho\beta\,\chi\big)^2}{2\beta} \label{pot}\\
		B & := d \Theta,  \qquad \Theta:=\varepsilon \hat{\omega} + \rho\mathcal A_S.  \label{field}
\end{align}
Notice that, by \eqref{potential-components},
$\varepsilon-\rho\beta\chi=\varepsilon+\rho\mathcal A(\partial_t)=\mathcal E_k$,
so that $V=-\mathcal E_k^2/(2\beta)$. Remarkably, the  total  energy of  the  particle  coincides with  the  quantity $\tfrac{1}{2} g(\dot{\gamma}, \dot{\gamma})$  (see \eqref{en}). This allows us to consider solutions parametrized by proper time and fixed charge-to-mass ratio.

\begin{theorem}\label{lorentz-lag}
		Let $(M= \R \times S,g)$ be a stationary spacetime, with $g$  as in \eqref{1a},
		and let $\rho, \varepsilon \in \R$. 
		\begin{enumerate}
			\def\labelenumi{(\alph{enumi})}
			\item If $\gamma(s) = (t(s),\sigma(s))$ is a solution of~\eqref{eqLorentz} with parameter $s$ such that 
			\begin{equation} \label{const}
				\varepsilon=\beta\big(\dot t-\hat\omega(\dot \sigma)+\rho\,\chi(\sigma)\big),
				\qquad 
				\frac{1}{2} g(\dot{\gamma}, \dot{\gamma}) = e,
			\end{equation}
			then $\sigma$ satisfies
			\begin{align}
				& \nabla^{h_0}_{\dot{\sigma}} \dot{\sigma} + \nabla^{h_0} V(\sigma)  =-  \big( \iota_{\dot{\sigma}} B \big)^{\sharp, h_0}, \label{sys} \\
				& \frac{1}{2} h_0(\dot{\sigma}, \dot{\sigma}) + V(\sigma) = e, \label{en}
			\end{align}
			where $V$ and $B$ are defined in \eqref{pot} and \eqref{field}, $\nabla^{h_0}$ is the Levi-Civita connection of $h_0$ and $\nabla^{h_0}V$ is the gradient of $V$ with respect to $h_0$, 
			$\iota_{\dot{\sigma}} B := B(\dot{\sigma}, \cdot)$  and ${}^{\sharp, h_0}$ denotes the musical isomorphism with respect to $h_0$.
			
			\item Conversely,  if $\sigma(s)$ is a solution of \eqref{sys}, \eqref{en}, 
			and $t$ solves
			\begin{equation} \label{en1}
				\dot t=\frac{\varepsilon}{\beta}+\hat\omega(\dot \sigma)-\rho\,\chi(\sigma),
			\end{equation}
			then $\gamma(s) = (t(s), \sigma(s))$ is a solution of~\eqref{eqLorentz} such that $\frac 1 2 g(\dot\gamma,\dot\gamma)=e$.
		\end{enumerate}
	\end{theorem}
	
	\begin{proof}
		Set 
		\[
		u:=\dot t-\hat\omega(\dot \sigma)
		\]
		so that the integrand of $\mathscr F$ in \eqref{eq:functional_f} is
		\begin{equation*}
			\frac12\,h_0(\dot \sigma,\dot \sigma)-\frac{\beta}{2}\,u^2-\rho\beta\chi\,u+\rho\,\mathcal A_S(\dot \sigma).
		\end{equation*}
		Let $(t_\epsilon,\sigma_\epsilon)$ be a variation of $\gamma$ with fixed endpoints and denote
		$\tau:=\partial_\epsilon t_\epsilon|_{\epsilon=0}$, $v:=\partial_\epsilon \sigma_\epsilon|_{\epsilon=0}$.
		Since  $\gamma$ is a critical point of $\mathscr F$, taking a variation with respect to $\sigma$ only and using the first equation in \eqref{const}, i.e. $\beta(u+\rho\chi)=\varepsilon$ (constant),  we get  (after integration by parts):
		\begin{equation}\label{weak}
			\mathscr F'(\gamma)[(0,v)]=-\!\int_a^b\!\!\Big(
			h_0(\nabla^{h_0}_{\dot \sigma}\dot \sigma,v)
			+\frac12 u^2\,d\beta(v)
			+\rho u\,d(\beta\chi)(v)
			-\varepsilon\,d\hat\omega(v,\dot \sigma)
			-\rho\,d\mathcal A_S(v,\dot \sigma)
			\Big)\,ds=0.
		\end{equation}
		From \eqref{pot} we get 
		\[
		dV
		=\frac{(\varepsilon-\rho\beta\chi)^2}{2\beta^{2}}\,d\beta
		+\rho\,\frac{\varepsilon-\rho\beta\chi}{\beta}\,d(\beta\chi),
		\]
		and since $u=\varepsilon/\beta-\rho\chi$, along $\gamma$ we get:
		\[
		dV=\frac{u^{2}}{2}\,d\beta+\rho\,u\,d(\beta\chi).
		\]
		Therefore, along $\gamma$ we have   $h_0(\nabla^{h_0}V,v)=\frac12 u^2\,d\beta(v)+\rho u\,d(\beta\chi)(v)$ and,
		as $v$ is arbitrary,  \eqref{weak} yields
		\[
		\nabla^{h_0}_{\dot \sigma}\dot \sigma+\nabla^{h_0}V(\sigma)
		=-(\iota_{\dot \sigma}B)^{\sharp,h_0},
		\]
		which is exactly \eqref{sys}. The remaining equation \eqref{en} follows trivially  from the second equation in \eqref{const}.
		
		For the converse, assume $\sigma$ solves \eqref{sys} and \eqref{en}, and let $t$ solve \eqref{en1}. Then
		$\varepsilon=\beta(\dot t-\hat\omega(\dot \sigma)+\rho\chi)$, i.e. both equations in  \eqref{const} hold.
		Moreover, reversing the computations above shows that \eqref{sys} is equivalent to
		\eqref{weak}, hence both  partial differentials  of $\mathscr F$ vanish and then  $\gamma$ is a critical point of $\mathscr F$, hence 	it solves \eqref{eqLorentz}, 	and \eqref{defe} follows from \eqref{en}.
	\end{proof}

	\subsection{Massive particle surfaces as totally geodesic surfaces}
	
	According to Definition~\ref{rhoepsMPS}, we are interested in solutions of \eqref{eqLorentz} having energy $\varepsilon$ and parametrized by proper time. 
	This means that we must set $e = -\tfrac{1}{2}$ in Theorem~\ref{lorentz-lag}.

	Moreover, we note that \eqref{sys} represents the equation of motion of a particle $\sigma$ moving under the action of the potential $V$ in \eqref{pot} and the magnetic field having potential  one-form 
	$\Theta:=\varepsilon \hat{\omega} + \rho \mathcal A_S$.
	
	We know that, under suitable assumptions on $V$ and $\Theta$, fixed-energy solutions of \eqref{sys} are pregeodesics with respect to a Jacobi--Randers metric (see Appendix~\ref{app} for details).
	More precisely, for a fixed pair $(\rho, \varepsilon)$, we consider the function $F$  defined by 
	\begin{equation} \label{finlerjac}
		F(x, y) = \sqrt{\left(-1 -2 V(x)\right) h_0|_x(y, y)} + \Theta_x(y).
	\end{equation}
	Taking into account the definition of $V$ in \eqref{pot}, the conditions \eqref{fincond} with  $e = -\tfrac{1}{2}$ become
	\begin{equation} \label{hyp}
		\frac{\big(\varepsilon-\rho\beta\,\chi\big)^2}{\beta} > 1 \;\mbox{on $S$} \quad\text{and}\quad
		\sup_{v \in T_xS \setminus \{0\}} \frac{|\Theta(v)|}{\sqrt{h_J|_x(v, v)}} < 1\; \mbox{for any $x \in S$},
	\end{equation}
	where $h_J$ is the Jacobi metric defined by 
	\begin{equation}\label{hJ}
		h_J =  (-1 - 2V ) h_0.
	\end{equation}	 
	In the following, given a Riemannian metric $h$ on $S$ and a one-form $\eta$ on $S$,  we  denote by  $\|\eta\|_h(x)$ the operator norm of $\eta_x$ with respect to $h_x$, $x\in S$; when  no  confusion can arise, we omit  the dependence on the point $x\in S$.
	
	More generally, we use the following definition
	\begin{definition}[Randers domain]\label{Randersdomain}
		Fix $(e,\rho,\varepsilon)\in\R^3$. Let $V$ and $\Theta$ be as in \eqref{pot}--\eqref{field} and set
		\begin{equation}\label{jac}
			h_e:=2(e-V)\,h_0 .
		\end{equation}
		The \emph{Randers domain}  of 
		\begin{equation}\label{finjac1}
			F_e(x,y):=\sqrt{h_e|_x(y,y)}+\Theta_x(y)
		\end{equation}
		is the open set
		\[
		\mathcal D_{e,\rho,\varepsilon}:=\Big\{x\in S:\ e>V(x)\ \text{and}\ \|\Theta\|_{h_e}(x)<1\Big\},
		\]
		on which  $F_e$ becomes  a  Finsler metric.
		
		In particular, in the proper-time case $e=-\tfrac12$ we write
		\[
		\mathcal D_{\rho,\varepsilon}:=\mathcal D_{-1/2,\rho,\varepsilon},
		\]
		and $F=F_{-1/2}$ is precisely the function in \eqref{finlerjac} with domain $T\mathcal D_{\rho,\varepsilon}$.
	\end{definition}

	\begin{remark} Notice that \eqref{hyp} is equivalent to $\mathcal D_{\rho,\varepsilon}=S$.
		Assumption \eqref{hyp} requires pointwise bounds  on  $S$; however, for the  characterization of a $(\rho,\varepsilon)$--MPS, it is enough to require the Randers condition only near $S_0$.	
		More precisely, the defining inequalities of $\mathcal D_{\rho,\varepsilon}$ are strict and depend continuously on $x$.
		Therefore, if they hold on $S_0$, then by continuity there exists an open neighborhood $U\supset S_0$ such that
		$U\subset\mathcal D_{\rho,\varepsilon}$, i.e. $F$ is a Randers metric on $U$.
		
		We notice that the two pointwise conditions in \eqref{hyp} can be compressed into a single strict inequality:
		\begin{equation*}
			1+2V+\|\Theta\|_{h_0}^2<0.
		\end{equation*}
		Indeed, $\|\Theta\|_{h_J}(x)<1$ is equivalent to $\|\Theta\|_{h_0}^2(x)<(-1-2V(x))$, which in particular forces $(-1-2V(x))>0$.
		Hence the Randers domain can be written as the open set
		\begin{equation}\label{Phi}
			\mathcal D_{\rho,\varepsilon}=\Big\{x\in S:\ \Phi_{\rho,\varepsilon}(x)<0\Big\},
			\qquad
			\Phi_{\rho,\varepsilon}(x):=1+2V(x)+\|\Theta\|_{h_0}^2(x).
		\end{equation}
		
Moreover, for a $\partial_t$-invariant hypersurface $\scrT=\R\times S_0$ one has $\kappa=K=\partial_t$ and $\kappa^2=\beta$. Consequently, at the proper-time level 	$e=-\tfrac12$, since $V=-\mathcal E_k^2/(2\beta)$,
\[
-\frac12>V\quad\Longleftrightarrow\quad\mathcal E_k^2>\beta=\kappa^2.
\]
Thus,  the  energy condition  in Theorem~\ref{thmumbMPS} (recall also Remark~\ref{degenerate}) is exactly the first 	pointwise condition defining the Jacobi--Randers domain.
		\end{remark}
		
\begin{remark}[Gauge transformation and Randers domain]
Let $\mathcal A'=\mathcal A+df$ 	be a stationary gauge transformation. Since $K=\partial_t$ and 		$\mathcal L_K(df)=0$, on the connected spacetime $M=\R\times S$ we can write $f(t,x)=ct+f_0(x)$ 		for some constant $c\in\R$ and some smooth function $f_0:S\to\R$. 	According to Remark~\ref{gauge-energy}, the corresponding energy level is $\varepsilon'=\varepsilon-\rho c$.	Writing $\mathcal A=\mathcal A_t\,dt+\alpha$ and $\mathcal A'=\mathcal A_t'\,dt+\alpha'$,	we have
$\mathcal A_t'=\mathcal A_t+c$ and $\alpha'=\alpha+df_0$. 
Since $\mathcal E_k=\varepsilon+\rho\mathcal A(K)$ is gauge invariant (Remark~\ref{gauge-energy}), the potential $V=-\mathcal E_k^2/(2\beta)$, as well as the Jacobi metric $h_e=2(e-V)h_0$, is gauge invariant. On the other hand, by \eqref{AS}, $\mathcal A_S'=\mathcal A_S+df_0+c\,\widehat\omega$, and then
$\Theta'	=	\varepsilon'\widehat\omega+\rho\mathcal A_S'=	\Theta+\rho\,df_0$, 
so that  $B'=d\Theta'=d\Theta=B$.
Thus the Jacobi metric $h_e$ and the magnetic field $B$ are gauge invariant, whereas the Randers representative changes according to 
$F_e'=F_e+\rho\,df_0$.
The Randers condition is not gauge invariant: even if $F_e$ is a Finsler metric, $F_e'$ is a Finsler metric only where $\|\Theta'\|_{h_e}<1$.
If both $F_e$ and $F_e'$ are Finsler metrics on a neighborhood of a submanifold, then their length functionals differ only by an endpoint term and they have the same oriented unparametrized geodesics; hence they have the same totally geodesic submanifolds. In general, however, only the
characterization in Theorem~\ref{riemannian-magnetic-MPS} remains 
meaningful independently of the gauge.
\end{remark}
Theorem \ref{jacprinc} can be restated as follows.
\begin{proposition}\label{lorjacprin}
		Fix $(\rho,\varepsilon)\in\R^2$ and set $e=-\tfrac12$.
		Let $\mathcal D_{\rho,\varepsilon}$ be the Randers domain in Definition~\ref{Randersdomain}, and let
		$F$ be given by \eqref{finlerjac} on $T\mathcal D_{\rho,\varepsilon}$.	If $\sigma:I\to\mathcal D_{\rho,\varepsilon}$ is a nonconstant geodesic of the Finsler manifold $(\mathcal D_{\rho,\varepsilon},F)$, then
		there exists a reparametrization of $\sigma$ which solves \eqref{sys} and has total energy \eqref{en} equal to $-\tfrac12$.
		Conversely, every solution $\sigma:I\to\mathcal D_{\rho,\varepsilon}$ of \eqref{sys} with total energy $-\tfrac12$
		is a pregeodesic of $(\mathcal D_{\rho,\varepsilon},F)$ (i.e.,  it admits a reparametrization such that it  becomes   a geodesic of $F$).
	\end{proposition}
We can now  characterize $(\rho,\varepsilon)$--MPS in terms of total geodesicity  of $S_0$ with respect to the Jacobi--Randers metric $F$.
		\begin{theorem}\label{mpschar}
		Let $(M=\R\times S,g)$ be a stationary spacetime with $g$ as in \eqref{1a}.
		Let $S_0$ be a smooth embedded hypersurface of $S$  and  consider 
		$\scrT=\R\times S_0$. Fix $(\rho,\varepsilon)\in\R^2$ and assume that $S_0\subset \mathcal D_{\rho,\varepsilon}$.	
		Then the following statements are equivalent:
		\begin{enumerate}
			\item $\scrT$ is a $(\rho,\varepsilon)$--MPS (Definition~\ref{rhoepsMPS});
			\item $S_0$ is totally geodesic in the Finsler manifold $(\mathcal D_{\rho,\varepsilon},F)$,
			where $F$ is given by \eqref{finlerjac} restricted to $T\mathcal D_{\rho,\varepsilon}$.
		\end{enumerate}
	\end{theorem}
	\begin{proof}
		Assume that $\scrT = \R \times S_0$ is a $(\rho, \varepsilon)$--MPS. For any  $x_0 \in S_0$, $v \in T_{x_0} S_0\setminus\{0\}$, let     $\sigma$ be the  geodesic of $F$ starting at $x_0$ with tangent vector $v$. By Proposition~\ref{lorjacprin}, there exists a reparametrization $\tilde{\sigma}$ of $\sigma$  solving \eqref{sys} and having total energy equal to $-\tfrac{1}{2}$. If $t$ is any solution of \eqref{en1}, then, by Theorem \ref{lorentz-lag}-(b),  $\gamma= (t, \tilde{\sigma})$ is a solution of \eqref{eqLorentz}, parametrized by proper time,  whose initial point lies in  $\scrT$ and initial velocity lies in $\R \times T_{x_0}S_0$.  Thus $\gamma $ is locally contained in $\scrT$ and consequently $\tilde{\sigma}$ (and hence $\sigma$) is locally contained in $S_0$, which is therefore totally geodesic with respect to $F$.
		
		Assume now that $S_0$ is totally geodesic with respect to the metric $F$. Since $S_0\subset\mathcal D_{\rho,\varepsilon}$, the identity $h_0(v,v)=-1-2V>0$ and \eqref{en1} give a $(\rho,\varepsilon)$-admissible tangent vector at every point of $\scrT$. Let then  $p_0 = (t_0, x_0) \in \scrT$ and take any $(\rho,\varepsilon)$-admissible tangent vector $u_0$; let $\gamma= (t, \sigma)$ be the solution of \eqref{eqLorentz} starting at $p_0$ with tangent vector $u_0$.  By Theorem \ref{lorentz-lag}-(a), $\sigma$ solves \eqref{sys} and has total energy equal to $-\tfrac{1}{2}$.  Hence, by Proposition~\ref{lorjacprin}, there exists a reparametrization $\tilde \sigma$ of $\sigma$  which is a geodesic of $F$. As its initial point lies on    $S_0$ and its initial velocity lies in $T_{x_0}S_0$,  $\tilde{\sigma}$ (and therefore $\sigma$) is locally   contained in $S_0$ and hence $\gamma$ is locally contained in $\scrT$. Therefore, $\scrT $ is a $(\rho, \varepsilon)$-MPS.
		\end{proof}
	\begin{remark}[Future and past orientations]
	As mentioned in Remark~\ref{futurepast}, for a timelike solution with fixed $(\rho,\varepsilon)$, the time 	orientation is determined by $\mathcal E_k=\varepsilon-\rho\beta\chi$. Indeed, 	the solution is future-pointing when $\mathcal E_k>0$ and past-pointing when $\mathcal E_k<0$. Notice that the Jacobi--Randers construction above 	does not require a choice of time orientation and applies in both cases.
	There is, however, an important difference from the null case. 	Reversing the orientation of a solution of the Lorentz force equation 	does not preserve the equation at fixed $\rho$. If $\gamma$ is a solution 		with charge-to-mass ratio $\rho$ and energy $\varepsilon$, then 	$\bar\gamma(s):=\gamma(-s)$ is a solution with charge-to-mass ratio $-\rho$ and energy $-\varepsilon$. 	Under the simultaneous 	change $(\rho,\varepsilon)\mapsto(-\rho,-\varepsilon)$,
 $V$ in \eqref{pot} is invariant while $\Theta$ is changed to $-\Theta$, 	and hence the Jacobi metric and the Randers domain are unchanged, while the 	Jacobi--Randers metric $F$ in \eqref{finlerjac} is replaced by its reverse metric. 	This holds in particular when $\rho=0$, so that future- and past-pointing timelike 	geodesics with energies $\varepsilon$ and $-\varepsilon$ are described by mutually reverse Jacobi--Randers metrics. This is the massive 	counterpart of the relation between the two Fermat metrics $F^\pm$ in \eqref{Fermat} in the null case.
	\end{remark}
	\begin{remark}
		A special case of Theorem \ref{lorentz-lag} concerns null geodesics, obtained by setting 
		$e = 0$ and $\mathcal{A} = 0$. Hence, it can be used to investigate photon surfaces 
		(see Definition \ref{photonsurface}). A null geodesic $\gamma = (t,\sigma)$ has  energy equal to 
		$ \varepsilon   = -\,g(\partial_t, \dot\gamma ) 
		= \beta \left(\dot t -\hat{\omega }  (\dot \sigma ) \right)$, where $\varepsilon>0$ (resp. $\varepsilon<0$) if and only if $\gamma$ is future-pointing (resp. past-pointing). Its projection $\sigma$ on $S$ is a solution of the system \eqref{sys}, where (by  \eqref{pot}, \eqref{field}) the potential and the magnetic field are
		$$ V=- \frac{\varepsilon^2}{2 \beta}, \quad B =d(\varepsilon \hat{\omega}).$$
		The associated Jacobi metric is $\tfrac{\varepsilon^2}{\beta}h_0= \varepsilon^2\hat{h}$, thus  the conditions in \eqref{fincond} hold and the solutions of \eqref{sys} with vanishing  total energy  \eqref{en} are pregeodesics of the Randers metric 
		$$F(x, y) = |\varepsilon|  \sqrt{\hat{h}_x (y,y)}+  \varepsilon  \hat{\omega}_x(y)  
		$$
		which equals  $\varepsilon F^+$ if $\varepsilon >0$ and $-\varepsilon F^{-}$ if $\varepsilon<0$, where $F^{\pm}$ are the Fermat metrics as in \eqref{Fermat}. 
		
		Since a null, future-pointing (resp.\ past-pointing) geodesic can be
		reparametrized so that $\varepsilon=1$ (resp.\ $\varepsilon=-1$),
		the analogue of Theorem~\ref{mpschar} for null geodesics leads to the
		following conclusion: a $\partial_t$-invariant hypersurface
		$\scrT=\mathbb R\times S_0\subset M$ is a photon surface if and only if
		$S_0$ is totally geodesic with respect to $F^+$, or equivalently with
		respect to $F^-$, as already shown in Theorem~\ref{thm:photonumb}.
		\end{remark}
Before stating the next characterization, let us point out a feature of the
fixed-energy reduction that will be crucial in its proof.  At the spacetime
level, if $\xi$ is a $(\rho,\varepsilon)$-admissible vector, then $-\xi$ is
in general not admissible at the same energy level, since in general 
$-g(K,-\xi)-\rho\,\mathcal A(K)\neq-g(K,\xi)-\rho\,\mathcal A(K)$.
After reduction to the spacelike slice, however, a different symmetry
appears.  Let  $v\in T_pS$ be $h_J$-unit.  The spatial velocities
$u_\pm:=\mp(1+2V(p))v$
can both be lifted to $(\rho,\varepsilon)$-admissible spacetime vectors by
setting
\[
\xi_\pm:=
\left(
\frac{\varepsilon}{\beta(p)}
+\widehat\omega_p(u_\pm)-\rho\chi(p),
\,u_\pm
\right).
\]
Indeed, both vectors have Lorentzian norm $-1$ and Killing energy
$\varepsilon$.  Notice that $\xi_-$ is in general not equal to $-\xi_+$:
reversing the spatial velocity is accompanied by a corresponding change of
the time component so as to remain on the same energy level.  Thus, although
the fixed-energy Lorentz-force dynamics is not invariant under
$\xi\mapsto-\xi$, its reduced initial data are symmetric under
$v\mapsto-v$.  This observation will allow us below to separate the part
which is quadratic in $v$, coming from the second fundamental form of
$S_0$ with respect to $h_J$, from the part which is linear in $v$, coming
from the magnetic field.
\begin{theorem}\label{riemannian-magnetic-MPS}
 Let $(M=\R\times S,g)$ be a stationary spacetime with $g$ as in  	\eqref{1a}. 
 Let $S_0$ be a smooth embedded hypersurface of $S$ and consider $\mathscr T=\R\times S_0$. Fix $(\rho,\varepsilon)\in\R^2$ and assume that $V<-1/2$ on $S_0$.  Then $h_J$ is Riemannian on a neighborhood of $S_0$, and  $\mathscr T$ is a $(\rho,\varepsilon)$--MPS if and only if
 \begin{enumerate}
 \item $S_0$ is totally geodesic with respect to $h_J$;
 \item $B(\nu,v)=0$, for every $v\in TS_0$, 	where $\nu$ is any local $h_J$-unit normal to $S_0$.
 \end{enumerate}
 \end{theorem}
 \begin{proof}
Let us set $\lambda:=-1-2V$. By continuity, $\lambda>0$ on an open neighborhood of $S_0$, and $h_J=\lambda h_0$ is Riemannian there. For a solution $\sigma=\sigma(s)$ of \eqref{sys} with total energy $e=-\tfrac12$, writing $u=d\sigma/ds$, equation~\eqref{en} gives $h_0(u,u)=\lambda$.
Hence $h_J(u,u)=\lambda^2$. Define a new parameter $\tau$ by $\frac{d\tau}{ds}=\lambda(\sigma(s))$.
Then the velocity with respect to $\tau$ is $v:=d\sigma/d\tau=u/\lambda$, and $h_J(v,v)=1$.

The Levi-Civita connections of $h_0$ and $h_J=\lambda h_0$ are related by (see, e.g. \cite[Prop.~7.29]{LeeRiem}):
 \[
 \nabla^{h_J}_X X= 	\nabla^{h_0}_X X 	+\frac{X(\lambda)}{\lambda}X 	-\frac{h_0(X,X)}{2\lambda}\nabla^{h_0}\lambda.
 \]
 Thus,
\[
\lambda^2\nabla^{h_J}_v v=\lambda^2\nabla^{h_0}_v v	+\lambda v(\lambda)v	-\frac12\nabla^{h_0}\lambda
	=\nabla^{h_0}_u u-\frac12\nabla^{h_0}\lambda=\nabla^{h_0}_u u+\nabla^{h_0}V,
\]
where in the second equality we used
\[
\nabla^{h_0}_u u=\nabla^{h_0}_{\lambda v}(\lambda v)
=\lambda^2\nabla^{h_0}_v v+\lambda v(\lambda)v,
\]
and in the last one $\lambda=-1-2V$.
Using \eqref{sys}, together with
\[
(\iota_{u}B)^{\sharp,h_0} 	= 	\lambda^2(\iota_vB)^{\sharp,h_J}, 	\]
 it follows that
 \begin{equation}\label{magnetic}
 \nabla^{h_J}_{d\sigma/d\tau} d\sigma/d\tau 	= 	-\bigl(\iota_{d\sigma/d\tau}B\bigr)^{\sharp,h_J}.
 \end{equation}
 Now assume that $\mathscr T$ is a $(\rho,\varepsilon)$--MPS.  Fix $p\in S_0$ and an $h_J$-unit vector $v\in T_pS_0$.
 Set $ u_0:=\lambda(p)v$. Since $h_J=\lambda h_0$ and $h_J(v,v)=1$, we have $h_0(u_0,u_0)=\lambda(p)$. 
 Hence
 \[
 \frac12 h_0(u_0,u_0)+V(p)=-\frac12.
 \]
 Fix $t_0\in\R$ and define
 \[
 \xi_0:=
 \left(
 \frac{\varepsilon}{\beta(p)}
 +\widehat\omega_p(u_0)-\rho\chi(p),
 \,u_0
 \right)
 \in T_{(t_0,p)}M.
 \]
 By the definition of $\lambda$ and $V$, $\lambda=-1-2V=\mathcal E_k^2/\beta-1$; thus, recalling \eqref{1a} and the definition of $\xi_0$, we obtain
 \[
 g(\xi_0,\xi_0)=-\frac{\mathcal E_k^2}{\beta}+h_0(u_0,u_0)=-1,
 \qquad
 -g(K,\xi_0)-\rho\mathcal A(K)=\mathcal E_k-\rho\mathcal A(K)=\varepsilon.
 \]
 Thus $\xi_0$ is a $(\rho,\varepsilon)$-admissible vector tangent to $\mathscr T=\R\times S_0$.  Let $\gamma=(t,\sigma)$ be the solution of the Lorentz force equation  with initial data $((t_0,p),\xi_0)$. Since $\mathscr T$ is a  $(\rho,\varepsilon)$--MPS, $\gamma$ locally remains in $\mathscr T$,  and hence $\sigma$ locally remains in $S_0$.  By Theorem~\ref{lorentz-lag}, $\sigma$ satisfies \eqref{sys} with total  energy $e=-\tfrac12$. After reparametrizing it as above,  we get a curve $\sigma=\sigma(\tau)$ satisfying   \eqref{magnetic}, whose initial velocity is  $v=u_0/\lambda(p)$.   Consequently, by uniqueness of the initial value problem for  \eqref{magnetic}, the solution of \eqref{magnetic} with initial data  $(p,v)$ locally remains in $S_0$.
 Taking the normal component of 	\eqref{magnetic} at $p$ gives $\Pi_J(v,v)=B(v,\nu)$,
 where $\nu$ is a local $h_J$-normal vector field to $S_0$ and $ \Pi_J(X,Y):=h_J(\nabla^{h_J}_X\nu,Y)$.
Repeating the above construction with $-v$ in place of $v$, we obtain another
$(\rho,\varepsilon)$-admissible vector tangent to $\mathscr T$, whose spatial
projection, after the Jacobi reparametrization, has initial velocity $-v$.
Hence, by the MPS property, the corresponding solution of
\eqref{magnetic} also locally remains in $S_0$. Taking its normal component
at $p$ gives
$\Pi_J(v,v)=\Pi_J(-v,-v)=B(-v,\nu)=-B(v,\nu)$.
Together with $\Pi_J(v,v)=B(v,\nu)$, this yields $\Pi_J(v,v)=0$ and $B(v,\nu)=0$.
Therefore, both $(1)$ and $(2)$ hold.
 	
 Conversely, assume $(1)$ and $(2)$. The construction of $\xi_0$ above shows that $(\rho,\varepsilon)$-admissible tangent vectors exist at every point of $\scrT$. Let $\nabla^{S_0}$ denote the Levi-Civita connection of the metric induced by $h_J$ on $S_0$. Since $S_0$ is totally geodesic, for tangent vector fields $X,Y$ on $S_0$ we have $\nabla^{h_J}_X Y=\nabla^{S_0}_X Y$. Moreover, condition $(2)$ implies that $\bigl(\iota_vB\bigr)^{\sharp,h_J}\in T_xS_0$ for every $x\in S_0$, $v\in T_xS_0$. Therefore \eqref{magnetic}  restricts to an equation on $S_0$. 
Hence, by uniqueness, every solution of this equation with initial data in $TS_0$ locally remains in $S_0$.

Now let $\gamma=(t,\sigma)$ be a solution of the Lorentz force equation
with $(\rho,\varepsilon)$-admissible initial data tangent to $\mathscr T$. By
Theorem~\ref{lorentz-lag}, $\sigma$ satisfies \eqref{sys} with total
energy $-\tfrac12$. After the Jacobi reparametrization described above,
the same unparametrized curve satisfies \eqref{magnetic} and has
initial velocity tangent to $S_0$. Hence its image locally lies in
$S_0$. Since reparametrization does not change the image of a curve,
the original projection $\sigma$ also locally lies in $S_0$.
Therefore $\gamma$ locally lies in
$\mathscr T=\R\times S_0$, and $\mathscr T$ is a
$(\rho,\varepsilon)$--MPS.
\end{proof}

Taking into account Theorems~\ref{mpschar} and \ref{riemannian-magnetic-MPS}, we can give examples of MPS in a standard stationary region of an exact solution of the Einstein--Maxwell equations.

\begin{example}\label{SomRay}
Let $a>0$ and set
\[
S:=\bigl\{(x,y,z)\in\R^3:x^2+y^2<a^{-2}\bigr\},	\qquad	M:=\R\times S.
\]
On $M$, consider
\[
\theta:=dt+a(x\,dy-y\,dx),\qquad	g=-\theta\otimes\theta+dx^2+dy^2+dz^2,
\]
and the stationary electromagnetic potential
\[
\mathcal A=\frac{a}{2}(x\,dy-y\,dx).
\]
These fields are the restriction to $M$ of the Som--Raychaudhuri solution of the Einstein--Maxwell equations
for a rigidly rotating charged dust. More precisely, we are considering the critical case $a^2=A^2$ of
\cite[\S3]{SomRay1968}, where the dust current is  $J=\frac{a^2}{2\pi}\partial_t$ while its mass density is $a^2/(4\pi)$. Hence, the dust has
charge-to-mass ratio $2$ and its  four-velocity $u=\partial_t$ satisfies
\[
\nabla_u u=0,\qquad \iota_u\mathcal F=0,	\qquad	\mathcal F=d\mathcal A=a\,dx\wedge dy,
\]
so that the Lorentz force acting on the  worldlines of the dust vanishes and these worldlines are timelike geodesics.
	
The vector field $K=\partial_t$ is complete and timelike. The metric induced 	on the slices of constant $t$ is
\[
g_0	=	dx^2+dy^2+dz^2-a^2(x\,dy-y\,dx)^2.
\]
Writing $r^2=x^2+y^2$, one sees that $g_0$ is positive definite for 	$r<a^{-1}$. Hence the above choice of $S$ gives a standard stationary spacetime with optical data \eqref{opticaldata} given by
\[
\beta=1,\qquad	\widehat\omega=-a(x\,dy-y\,dx),\qquad	h_0=dx^2+dy^2+dz^2.
\]
Using~\eqref{potential-components}--\eqref{AS}, we also obtain
\[
\chi=0,\qquad\mathcal A_S=\frac{a}{2}(x\,dy-y\,dx).
\]
Fix $\varepsilon\in\R$ with $\varepsilon^2>1$ and set $\rho=2\varepsilon$. Then, from~\eqref{pot}--\eqref{field},
\[
V=-\frac{\varepsilon^2}{2},	\qquad 	\Theta	=	\varepsilon\widehat\omega+\rho\mathcal A_S	=0.
\]
Consequently, from \eqref{finlerjac} and \eqref{hJ}
\[
h_J=(\varepsilon^2-1)h_0,\qquad	F(v)=\sqrt{(\varepsilon^2-1)h_0(v,v)},\qquad	B=d\Theta=0.
\]
Thus, the Jacobi--Randers metric reduces to a constant multiple of the Euclidean norm. Let $P\subset\R^3$ be any  plane such that $P\cap S\neq\emptyset$, and set 	$S_0:=P\cap S$.  	Then $S_0$ is totally geodesic with respect to both $h_J$ and $F$, and  by 	Theorem~\ref{mpschar} (or by Theorem~\ref{riemannian-magnetic-MPS}, since  $B=0$ and condition $(2)$ is trivially satisfied) we have that $\mathscr T:=\R\times S_0$ 	is a $(2\varepsilon,\varepsilon)$--MPS. 
	
Let us now examine the extrinsic geometry of these hypersurfaces.
Let $\nu$ be a Euclidean unit normal to $P$ and extend it to a constant vector field on $S$ (notice that replacing $\nu$ by $-\nu$ changes simultaneously the signs of $n$, $w$, $\Pi$, and $\iota_n\mathcal F$, so that all the conclusions below are independent of the choice of $\nu$).
 We denote by the same symbol $\nu$ its $t$-independent lift to $M=\R\times S$, tangent to the slices of constant $t$, and set
\[
n:=\widehat\omega(\nu)\,\partial_t+\nu.
\]
Since $\theta=dt-\widehat\omega$, we have $\theta(n)=0$. Hence, for $X=\tau\partial_t+v\in T\mathscr T$, $v\in TS_0$, 
\[
g(n,X)=h_0(\nu,v)=0, \qquad g(n,n)=h_0(\nu,\nu)=1.
\]
Thus, along $\mathscr T$, $n$ is a unit normal to $\mathscr T$. Therefore, by Cartan's formula,
\[
\mathcal L_n\theta=\iota_n d\theta+d(\theta(n))=\iota_n d\theta.
\]
Since $d\theta=2a\,dx\wedge dy$, and only the spatial component $\nu$ of $n$ contributes to the contraction,
we obtain
\[
\mathcal L_n\theta=2a(\nu_x\,dy-\nu_y\,dx)
=:w;
\]
since $n$ has constant spatial components we also have $\mathcal L_nh_0=0$, so that 
\[
\mathcal L_ng=-(w\otimes\theta+\theta\otimes w).
\]
Moreover, $\kappa=K=\partial_t$, and $\kappa^\flat=-\theta$. Since for all vector fields $X, Y$ tangent to $\mathscr T$ we have  $g(n, Y) = 0=g(n, [X,Y])$, differentiating yields $g(\nabla_X n, Y) = -g(n, \nabla_X Y) = -g(n, \nabla_Y X) = g(\nabla_Y n, X)$. Moreover $(\mathcal L_n g)(X,Y) = g(\nabla_X n, Y) + g(X, \nabla_Y n)$,  thus 
$\Pi(X,Y) = \frac{1}{2}(\mathcal L_n g)(X,Y)$. Hence,
\begin{equation}\label{PiSR}
\Pi	=	\frac12\bigl(w\otimes\kappa^\flat+\kappa^\flat\otimes w	\bigr)\big|_{T\mathscr T\times T\mathscr T},
		\qquad	w(\kappa)=0.
\end{equation}
Thus $\Pi$ has only mixed components with respect to the orthogonal 	decomposition
$T\mathscr T=D\oplus\operatorname{span}\{\kappa\}$,	where $D=\kappa^\perp$.

Notice that $w=0$ if and only if $\nu=\pm\partial_z$,	namely precisely when $P$ is orthogonal to the rotation axis. In this case 	$\Pi=0$, and hence $\mathscr T$ is totally geodesic. 
	
Assume instead that $P$ is not orthogonal to the rotation axis. Then $w\neq0$, and
\[
\iota_n\mathcal F	=	a(\nu_x\,dy-\nu_y\,dx)	=	\frac12 w.
\]

Since $\rho=2\varepsilon\neq0$ and $\iota_n\mathcal F=	\frac12 w$ is not proportional	to $\kappa^\flat=-\theta_\scrT$, where $\theta_\scrT$ is the one-form induced by $\theta$ on $\scrT$, Proposition~\ref{quasiumbilicequiv} shows that 	$\mathscr T$ is not quasi-umbilical. Moreover, by~\eqref{PiSR},
$\Pi$ is not proportional to $g_{\mathscr T}$, so $\mathscr T$ is not a
photon surface either.  	
\end{example}
\begin{remark}\label{SomRay-remarks}
The restriction $x^2+y^2<a^{-2}$ selects the standard stationary 	region of the Som--Raychaudhuri spacetime relevant to our setting. Indeed, 	in the full spacetime the circles with constant $t$, $z$ and
$x^2+y^2=r^2>a^{-2}$ are timelike. At $r=a^{-1}$ the metric induced on the slices of constant $t$ becomes
degenerate (indeed, in cylindrical coordinates it takes the form $g_0=dr^2+dz^2+r^2(1-a^2r^2)d\varphi^2$).

There is also a direct relation between the threshold energy level discussed in 	Remark~\ref{degenerate} and the matter source of the Som--Raychaudhuri 	solution.  Since 	$\mathcal A(K)=0$ we have $\mathcal E_k=\varepsilon$, for any charge-to-mass ratio $\rho$ on  every  timelike hypersurface.  As $g(K,K)=-1$ we have that  on a $K$-invariant hypersurface $\kappa^2=1$. Thus if $\varepsilon =1$, we have $\mathcal E_k^2=\kappa^2$ and    the unique $(\rho,1)$-admissible vector at each point is $u=K$. Since $\nabla_KK=0$ and $\iota_K\mathcal F=0$,
the integral curves of $K$,  	solve the Lorentz force equation for any charge-to-mass ratio $\rho$. Consequently, every 	$K$-invariant timelike hypersurface in the standard stationary region 	above is a  $(\rho,1)$--MPS, for any $\rho\in\R$, with the threshold energy level $\mathcal E^2_k=1=\kappa^2$.
	
Thus the Som--Raychaudhuri solution exhibits both regimes: the  family of 	$(2\varepsilon,\varepsilon)$--MPS with $\varepsilon^2>1$, described above, and the $(\rho,1)$ family of $K$-invariant hypersurfaces at the threshold energy level.
\end{remark}
\subsection{Solutions to the Lorentz force equation connecting a point to a line}
A further consequence of Theorems \ref{lorentz-lag} and \ref{jacprinc} concerns the existence of solutions $\gamma$ of \eqref{eqLorentz}, with a fixed charge-to-mass ratio $\rho$, in a standard stationary spacetime. We can consider two types of solutions:  solutions connecting a point $(t_0, p_0) \in M$ to a line $l :\R \to  M$, 
$u  \mapsto (u, p_1)$, $p_1\in S$ and solutions whose projection  $\sigma$ on $S$ is periodic.  The first  type of solutions for  the geodesic equation has been widely studied  in static and stationary spacetimes (see \cite{bgs,ger, bg}, and the seminal paper \cite{FGM95} where  variational methods were first used  in the study of  lightlike geodesics connecting a point to a flow line of the Killing vector field in a standard stationary spacetime, see also \cite{Masiello1994}).
	
Even though the physically relevant case is the one where $\gamma$ is parametrized with respect to proper time (cf. \cite{CapMin2004,Min2003,MinSan2006}) we can achieve more flexibility by looking at solutions having prescribed values of
the energy $\varepsilon$ and the constant $g(\dot\gamma, \dot\gamma) $ (leaving the physically important case to Remark~\ref{physicalsolo}).
In fact, thanks to Theorem~\ref{lorentz-lag}, it suffices to prove the existence of 
solutions of \eqref{sys} joining $p_0$ to $p_1$ or periodic with fixed total energy. In turn, these solutions can be obtained as 	pregeodesics of the Jacobi--Randers metric \eqref{finjac1}.
When $(S, h_0)$ is complete and    
\begin{equation} \label{hyp2}
		\inf_{S} \frac{\big(\varepsilon-\rho\beta\,\chi)^2}{\beta} > - 2e  \quad\text{and}\quad
		C:=   \sup_{x \in S} \|\Theta\|_{h_e}(x) < 1,
\end{equation}
	then   \eqref{finjac1} is a complete Randers metric on $S$.
	In fact,  observe that \eqref{hyp2} implies $\mathcal D_{e,\rho,\varepsilon}=S$ and, moreover,
	the uniform bound $C<1$ yields forward and backward completeness of $(S,F_e)$ (cf. Remark~\ref{fincompl}).

	Let us define  the function $U : S \rightarrow \R$ 	as 
	\begin{equation} \label{U}
		U= \frac{\varepsilon-\rho\beta\,\chi}{2e\beta + \big(\varepsilon-\rho\beta\,\chi)^2}.
	\end{equation}
	\begin{theorem} \label{pointline}
		Let $(M= \R \times S,g)$ be a stationary spacetime, with $g$ as in \eqref{1a} such that $S$ is connected  and $(S,  h_0)$ is complete. Let $\rho, \varepsilon, e \in \R$ such that \eqref{hyp2} holds.  Then
		\begin{enumerate}
			\def\labelenumi{(\alph{enumi})}
			\item for any $(t_0, p_0) \in M$ and any $p_1\in S$ with $p_0\neq p_1$, letting  $l :\R \to  M$ be the timelike line $l(u)=(u,p_1)$,		a solution $\gamma$ of \eqref{eqLorentz} exists joining $ (t_0, p_0 )$ to $l$ and satisfying \eqref{const};
			\item if $S$ is not contractible,  then for $(t_0,p_0)$, $p_1$ as in (a), there exists a sequence $(\gamma_n= (t_n, y_n))_n$ of solutions as in (a); their  arrival times $u_{\gamma_n}$ on $l$   diverge  to $+\infty$ provided  that $U$ satisfies the condition:
			\begin{equation}\label{dis}
			\delta:=	\inf_{x\in S}\bigl(U(x)-\|\widehat\omega\|_{h_e}(x)\bigr)>0;
				\end{equation}
			\item if $S$ is compact,  a  solution $\gamma=(t,y)$ of \eqref{eqLorentz} exists satisfying \eqref{const} and such that 
			$y$ is a nonconstant periodic curve on $S$.
		\end{enumerate}
	\end{theorem}
	\begin{proof} 
		Let us consider the functional
		\[
		J(\sigma) = \frac{1}{2} \int_0^1 F_e^2(\sigma,\dot{\sigma})\, ds
		\]
		defined on the manifolds $\Omega_{p_0,p_1}(S)$ or $\Lambda(S)$ of curves $\sigma$ in $S$ parametrized on the 
		interval $[0,1]$ having $H^1$-regularity, 
		i.e.  $\sigma$ is absolutely continuous and  
		$\int_0^1 h_0(\dot{\sigma},\dot{\sigma})\, ds<+\infty$, and satisfying respectively the boundary conditions: for statements $(a)$ and $(b)$, $\sigma$ connects the points $p_0$ and $p_1$; for statement $(c)$, $\sigma(0)=\sigma(1)$. 
		
		The properties of $J$ have been studied in \cite{cjm}. More precisely, a curve $\sigma \in \Omega_{p_0,p_1}(S)$ or $\sigma\in \Lambda(S)$ is a geodesic of the Finsler manifold $(S, F_e)$ with 
		$F_e(\sigma,\dot{\sigma})$ nonzero and constant if and only if it is a nonconstant critical 
		point of $J$. 
		
		Let us prove statements $(a)$ and $(b)$. Under the completeness assumption for $h_0$ and condition \eqref{hyp2}, 
		a critical point of $J$ on $\Omega_{p_0,p_1}(S)$ exists; moreover, if $S$ is not contractible, there exist 
		infinitely many geodesics $\sigma_n$ such that $J(\sigma_n)\to +\infty$ 
		(see \cite[Prop.~3.1]{cjm}).  Each critical point $\sigma : [0,1] \to S$ of $J$ can be 
		reparametrized as in  Theorem \ref{jacprinc}, yielding a new curve $y : [0,T] \to S$ joining $p_0$ to $p_1$, 
		which is a solution of \eqref{sys} with total energy $e$. 
		Finally, by \eqref{en1}, it suffices to define 
		$t : [0,T] \to \mathbb{R}$ by
		\begin{equation} \label{1}
			t(s) = t_0 + \int_0^s 
			\left( \frac{\varepsilon}{\beta (y)} 
			+ \hat{\omega}(\dot{y})  - \rho \chi(y)\right)\, d\tilde s  
		\end{equation}
		to obtain a solution $\gamma(s) = (t(s), y(s))$ of \eqref{eqLorentz} having the required properties.
		It remains to prove that the sequence of arrival times $u_{\gamma_n}$ in (b) diverges to $+\infty$.  Let $\sigma_n : [0,1] \rightarrow S$ be a sequence of critical points of $J$ such  that $J(\sigma_n) \rightarrow + \infty$. Setting the constants $c_n = F_e (\sigma_n, \dot{\sigma}_n)$, we have $2 J(\sigma_n) =c^2_n$,  hence $c_n \rightarrow + \infty$. Moreover, by \eqref{finjac1} and \eqref{hyp2},
		$$
		c_n = \int_0^1 F_e (\sigma_n, \dot{\sigma}_n) ds \leq (1+C) \int_0^1 \sqrt{h_e (\dot{\sigma}_n ,\dot{\sigma}_n ) }ds.
		$$
		Thus
		\begin{equation}\label{2a}
			\int_0^1 \sqrt{h_e (\dot{\sigma}_n ,\dot{\sigma}_n ) } ds \rightarrow + \infty.
		\end{equation}
		By Theorems \ref{jacprinc} and \ref{lorentz-lag},  the curve $\gamma_n= (t_n, y_n)$ is a solution  of \eqref{eqLorentz} with energy   $\varepsilon$ and   $ \tfrac{1}{2} g(\dot{\gamma_n}, \dot{\gamma_n})=e$  provided that  
		$y_n :[0, T_n] \rightarrow S$ is defined by  $y_n (s) = \sigma_n ( \theta_n (s))$, where $\theta_n : [0, T_n]\rightarrow [0,1]$ satisfies
		$$
		\dot\theta_n(s)=\frac{2\big(e-V(\sigma_n(\theta_n(s)))\big)}{\sqrt{h_e\big(\dot\sigma_n(\theta_n(s)),\dot\sigma_n(\theta_n(s))\big)}}
		$$
		and $t_n$ is defined as in \eqref{1}, with $y_n$ in place of $y$.
		Hence, recalling \eqref{pot} and \eqref{U}, the arrival times can be written as
		\begin{align} \label{at}u_{\gamma_n} = t_n (T_n)  &  = t_0 + \int_0^{T_n}  
			\left( \frac{\varepsilon}{\beta (y_n)} - \rho \chi(y_n)\right) \, ds 
			+ \int_0^{T_n} \hat{\omega}(\dot{y}_n) ds  \nonumber \\
			& = t_0 +\int_0^1   U(\sigma_n)\sqrt{h_e (\dot{\sigma}_n ,\dot{\sigma}_n )}ds + \int_0^{1} \hat{\omega}(\dot{\sigma}_n) ds .
		\end{align} 
		By \eqref{dis}, \eqref{2a} and \eqref{at} 
		$$u_{\gamma_n} \geq t_0 + \delta \int_0^1 \sqrt{h_e (\dot{\sigma}_n ,\dot{\sigma}_n ) } ds \to+\infty.$$
		
		Finally, the Finsler version of the Lyusternik--Fet theorem provides a nonconstant closed geodesic $\sigma$ on the compact manifold $S$. The positive reparametrization factor in Theorem~\ref{jacprinc} is periodic along $\sigma$, so the resulting solution $y$ of \eqref{sys} is periodic, with matching velocities at the endpoints of a period. Lifting $y$ by \eqref{en1} proves $(c)$. 
	\end{proof}
	\begin{remark}\label{physicalsolo}
		The physically relevant value of the constant $e$ is $-\frac 1 2$. Thus, recalling \eqref{finlerjac} and \eqref{hyp2}, Theorem \ref{pointline}-(a) holds if $h_0$ is complete and
		$$ \inf_{S} \frac{\big(\varepsilon-\rho\beta\,\chi\big)^2}{\beta}>1\quad\text{and}\quad
		\sup_{v \in TS \setminus \{0\}} \frac{|\Theta(v)|}{\sqrt{h_J(v, v)}} < 1 $$
		while for statement $(c)$  it is enough that \eqref{hyp} holds.
	\end{remark}
	In the next proposition we give a condition ensuring that the Randers domain $\mathcal D_{\rho,\varepsilon}$ 
	in  Definition~\ref{Randersdomain} coincides with $S$ for large values of  $|\varepsilon|$.
	Let  $a:S\to \R$ be the function 
	\[a(x)=\|\hat\omega\|_{h_0}^2(x)-\frac1{\beta(x)};\] 
	we notice that from \eqref{opticaldata}
	\[
	\frac{\hat\omega(v)^2}{h_0(v,v)}<\frac1\beta, \qquad \forall v\in T_x S\setminus\{0\},
	\]
	which implies that $a(x)<0$, for all $x\in S$.
	Let us set
	\[
	a_-:=\sup_S a,\qquad \phi_+:=\sup_S|\phi|, \qquad \ A_+:=\sup_S\|\mathcal A_S\|_{h_0}.
	\]
	\begin{proposition}\label{global-high-energy}
		Fix $\rho\in\R$ and  assume that $\phi_+, A_+$ are finite and  $a_-<0$.	
		Then there exists $\varepsilon_0=\varepsilon_0(\rho,\beta,\hat\omega,\phi, A)\geq 0$ such that,
		for every $\varepsilon\in\R$ with $|\varepsilon|>\varepsilon_0$, one has
		\[
		\Phi_{\rho,\varepsilon}(x)<0,\qquad\forall x\in S,
		\]
		where $\Phi_{\rho,\varepsilon}$ is defined in \eqref{Phi}. Hence the Randers domain is global, i.e.\ $\mathcal D_{\rho,\varepsilon}=S$, and \eqref{finlerjac}
		defines a Randers metric on $TS$.
	\end{proposition}
	\begin{proof}
		Taking into account \eqref{pot}--\eqref{field}, we have  for every $x\in S$
		\[		\Phi_{\rho,\varepsilon}
		=1-\frac{(\varepsilon-\rho\beta\chi)^2}{\beta}+\|\varepsilon\hat\omega+\rho\mathcal A_S\|_{h_0}^2.\]
		Since the operator norm $\|\cdot\|_{h_0}$ is  the norm induced on covectors by the dual metric $h_0^{-1}$ we have 
		\[
		\|\Theta\|_{h_0}^2
		=h_0^{-1}(\Theta,\Theta)
		=\varepsilon^2\|\hat\omega\|_{h_0}^2
		+2\rho\varepsilon\,h_0^{-1}(\hat\omega,\mathcal A_S)
		+\rho^2\|\mathcal A_S\|_{h_0}^2.
		\]
		As $\mathcal A_S=h_0(A,\cdot)$, we have $h_0^{-1}(\hat\omega,\mathcal A_S)=\hat\omega(A)$.	
		Hence
		\[\Phi_{\rho,\varepsilon}=1+\Big(\|\hat\omega\|_{h_0}^2-\frac1{\beta}\Big)\varepsilon^2
		+2\rho\big(\chi+\hat\omega(A)\big)\varepsilon
		+\rho^2\big(\|\mathcal A_S\|_{h_0}^2-\beta\chi^2\big).
		\]
		Recalling \eqref{chi},  $\chi+\hat\omega(A)=\phi$ and we obtain
		\begin{equation*}
			\Phi_{\rho,\varepsilon}
			=1+a\,\varepsilon^2+2\rho\,\phi\,\varepsilon+\rho^2\big(\|\mathcal A_S\|_{h_0}^2-\beta\chi^2\big).
		\end{equation*}
		Dropping the nonpositive term $-\rho^2\beta\chi^2$ 	gives
		\[
		\Phi_{\rho,\varepsilon}\le a_-\varepsilon^2+2|\rho|\phi_+|\varepsilon|+\big(1+\rho^2A_+^2\big).
		\]
		Therefore $\Phi_{\rho,\varepsilon}<0$ on $S$ as soon as
		\[
		a_-\varepsilon^2+2|\rho|\phi_+|\varepsilon|+\big(1+\rho^2A_+^2\big)<0.
		\]
		Set 
		\[
		\varepsilon_0:=		-\frac{|\rho|\,\phi_+ + \sqrt{\rho^2 \phi_+^2 - a_-\big(1+\rho^2 A_+^2\big)}}{a_-}.
		\]
		Since $a_-<0$, this is the positive root of 	$a_-r^2+2|\rho|\phi_+r+\big(1+\rho^2A_+^2\big)=0$. Hence the preceding polynomial is negative for every  $r>\varepsilon_0$ and therefore  $\Phi_{\rho,\varepsilon}<0$ on $S$ whenever $|\varepsilon|>\varepsilon_0$. Thus,   $\mathcal D_{\rho,\varepsilon}=S$.
	\end{proof}
	As a consequence of Theorem~\ref{pointline}-(c) and Proposition~\ref{global-high-energy} we immediately get the following result.
	\begin{corollary}\label{closedlarge}
		Fix $\rho\in\R$ and assume $S$ is compact.  Then there exists  $\varepsilon_0$ as in Proposition~\ref{global-high-energy} such 
		that  for every $\varepsilon\in\R$, with $|\varepsilon|>\varepsilon_0$, there exists a timelike solution
		$\gamma_\varepsilon=(t_\varepsilon,\sigma_\varepsilon)$ of \eqref{eqLorentz} parametrized by proper time
		with energy  $\varepsilon$, such that 	$\sigma_\varepsilon$ is a nonconstant periodic curve in $S$.
	\end{corollary}
	\begin{remark}
		Choosing a sequence of distinct energy levels $\varepsilon_k$ with 		$|\varepsilon_k|>\varepsilon_0$ yields infinitely many proper-time parametrized 		solutions with periodic projection, distinguished by their conserved Killing 	energies. Notice that their projected curves need not have distinct images in 	$S$.
	\end{remark}
	
	\begin{corollary}\label{closedMPS}
		Let $(M=\R\times S,g)$ be a stationary spacetime with $g$ as in \eqref{1a} and fix $(\rho,\varepsilon)\in\R^2$. Let $\scrT=\R\times S_0$ be a $(\rho,\varepsilon)$--MPS, where $S_0$ is a compact embedded hypersurface of $S$ such that $S_0\subset\mathcal D_{\rho,\varepsilon}$. Then there exists a timelike  solution $\gamma=(t,y)$ of \eqref{eqLorentz}, parametrized by proper time and with energy $\varepsilon$, which is contained in $\scrT$ and such that $y$ is a nonconstant periodic curve in $S_0$.
	\end{corollary}
	\begin{proof}
		Let $F$ be the Jacobi--Randers metric \eqref{finlerjac} on $\mathcal D_{\rho,\varepsilon}$ and let $F_0:=F|_{TS_0}$, which is  also a  Randers metric on $S_0$. By the Finsler version of the Lyusternik--Fet theorem, applied to a connected component of the compact manifold $S_0$, there exists a nonconstant closed geodesic $\sigma\colon\R\to S_0$ of $F_0$. By Theorem~\ref{mpschar}, $S_0$ is totally geodesic in $(\mathcal D_{\rho,\varepsilon},F)$, hence $\sigma$ is also a geodesic of $F$. By Proposition~\ref{lorjacprin} and Theorem~\ref{jacprinc}, there exists an orientation-preserving reparametrization $y$ of $\sigma$ which solves \eqref{sys} with total energy $-\tfrac12$. As in the proof of Theorem~\ref{pointline}-$(c)$, the reparametrization factor is periodic along $\sigma$, so that $y$ is a nonconstant periodic curve whose image is contained in $S_0$. Finally, defining $t$ by \eqref{en1}, Theorem~\ref{lorentz-lag}-$(b)$ shows that $\gamma=(t,y)$ is a solution of \eqref{eqLorentz} with energy $\varepsilon$ and $g(\dot\gamma,\dot\gamma)=-1$. Since $y$ takes values in $S_0$, $\gamma$ is contained in $\scrT=\R\times S_0$.
	\end{proof}
	\begin{remark}
		The time component of the solution $\gamma$ in Corollary~\ref{closedMPS} is not periodic: if $T>0$ is a period of $y$, then  $\gamma(s+T)=\big(t(s)+\Delta,\,y(s)\big)$ with $\Delta:=t(T)-t(0)$. Thus $\gamma$ is periodic up to a translation along the flow of $K=\partial_t$. Notice also that, unlike Theorem~\ref{pointline}, Corollary~\ref{closedMPS} requires neither compactness of $S$ nor completeness of $h_0$.
	\end{remark}
	\section*{Acknowledgments}
	\noindent E.C. and A.M. are  partially supported   by  MUR under the Programme ``Department of Excellence'' Legge 232/2016  (Grant No. CUP - D93C23000100001).\\
	All authors  are partially supported  by  INdAM - GNAMPA and in particular E.C. and A.M. by  ``INdAM - GNAMPA Project''  CUP E53C25002010001.

	\appendix
	\section{Jacobi--Randers metrics for non-relativistic electromagnetic systems}\label{app}
	
	Let $(N,h)$ be a connected Riemannian manifold  with Levi-Civita connection $\nabla$ and let $V$ be a smooth function on $N$.  The motion of a particle on $N$ under the   potential $V$ is governed by the equation
	$$\nabla_{\dot{\sigma}}\dot{\sigma} + \nabla V(\sigma)  =0$$
	where $\nabla V$ denotes the gradient of $V$ with respect to $h$. 
	The classical Maupertuis-Jacobi principle states that  the solutions of the above equation with constant energy
	\begin{equation} \label{toten}
		e= \frac{1}{2}h (\dot{\sigma}, \dot{\sigma}) + V(\sigma)
	\end{equation}
	are, on the region $e>V$ and up to orientation-preserving reparametrization, geodesics for the Jacobi metric $h_e=2(e-V)h$.
	
	When a particle also moves  under the action of a magnetic field, described by a closed  two-form   $\mathcal B$  on $N$, the   equation of motion becomes 
	\begin{equation} \label{mag}
		\nabla_{\dot{\sigma}}\dot{\sigma} + \nabla V(\sigma) =-\big(\iota_{\dot \sigma} \mathcal B\big)^\sharp   
	\end{equation}
	where   $\iota_{\dot \sigma} \mathcal B:=\mathcal B(\dot \sigma,\cdot)$ and ${}^\sharp$ is the musical isomorphism ${}^\sharp:T^*N\to TN$ associated with $h$. Since $\mathcal B$ is antisymmetric,  the total energy of a solution $\sigma$ remains equal to $e$ as defined  in \eqref{toten}. As seen in the first part of the proof of Theorem~\ref{riemannian-magnetic-MPS} in the case $e=-1/2$, at fixed total energy $e$, equation \eqref{mag} and   the Euler--Lagrange equations of  the length functional associated with $F_e=\sqrt{h_e}+\Theta$, where $\Theta$ is a local potential one-form of $\mathcal B$,  describe the same unparametrized curves	(see  also  \cite{mest2016}). 	
	In the remainder of the appendix, assume that $\mathcal B=d\Theta$ for a globally defined smooth one-form $\Theta$ on $N$. For a general closed two-form, the construction applies on any open set admitting such a primitive. We shall prove that the solutions of \eqref{mag} at energy $e$ coincide, up to orientation preserving reparametrization, with the geodesics of $F_e$ wherever it is a Randers metric (see also \cite{mest2016}). It is well known that $F_e$ defines a Finsler metric on $N$ of Randers type, 	provided that $h_e$ is Riemannian and $F_e$ is positive away from the zero section. 	Thus,  $F_e$ is a Randers metric on $N$ if and only if
	\begin{equation}\label{fincond}
		e > V(x) 
		\quad\text{and}\quad
		\sup_{v \in T_xN \setminus \{0\}} 
		\frac{|\Theta_x(v)|}{\sqrt{\,h_e|_x(v,v)\,}} < 1, 
		\quad 
		\forall x \in N.
	\end{equation}
	Indeed, \cite[Sect. 11.1]{bcs} shows that the second condition in \eqref{fincond} 
	not only ensures that $F_e$ is positive, but also that its square is  fiberwise strongly convex.
	\begin{remark} \label{fincompl}
	Assume $(N,h)$ is complete.	When the stronger conditions 
		\begin{equation} \label{fincond1}
			e>\sup_N V \quad\text{and}\quad
			\sup_{v\in TN\setminus\{0\}}\frac{|\Theta(v)|}{\sqrt{\,h_e(v,v)\,}}<1 
		\end{equation}
		hold, the Finsler manifold $(N, F_e)$  is forward and backward complete. Indeed,   the first condition in \eqref{fincond1} guarantees that $(N,h_e) $ is complete. Hence,  under the second condition, the same holds for $F_e$ (see \cite[Remark 4.1]{cjm} for details).
	\end{remark}
	\begin{remark}\label{geos}
	When  \eqref{fincond} holds, a direct computation of  the Euler--Lagrange equations for the action functional of the Lagrangian $F_e$  shows that the geodesics $\sigma$ of $F_e$  admit a parametrization such that they satisfy   
		\begin{equation} \label{eul}
		\nabla_{\dot{\sigma}} \left(\sqrt{ \frac{2 (e-V(\sigma))}{h(\dot \sigma , \dot \sigma)}}\dot{\sigma} \right) +  \sqrt{\frac{ h(\dot \sigma , \dot \sigma)}{ 2(e-V(\sigma))}} \nabla V (\sigma)  = -\big(\iota_{\dot \sigma} \mathcal B\big)^\sharp.
	\end{equation} 
\end{remark}
In the next theorem, we shall provide a proof of the relation between the geodesics of $F_e$ and the fixed-energy solutions of \eqref{mag}, 
which extends to the magnetic case the proof of the Maupertuis--Jacobi principle given in \cite{ben84}.
\begin{theorem}\label{jacprinc}
		Assume \eqref{fincond}. Every nonconstant geodesic of $F_e$ in \eqref{finjac1} admits an orientation preserving reparametrization that solves \eqref{mag} with total  energy \eqref{toten} equal to $e$. Conversely, every solution of \eqref{mag} with total energy $e$	 can be reparametrized to become  a geodesic of  $F_e$. 
	\end{theorem}
		\begin{proof}
		Let $\sigma:[a,b] \rightarrow N$ be a nonconstant	 geodesic of $F_e$. By Remark~\ref{geos}, $\sigma$ admits a parametrization such that   \eqref{eul} holds. Let us rewrite \eqref{eul} as:
		\begin{equation}\label{eul1}
			\nabla_{\dot{\sigma}} \left( \frac{\dot{\sigma}}{f} \right) +  f \nabla V (\sigma)  = -\big(\iota_{\dot \sigma} \mathcal B\big)^\sharp 
		\end{equation}
		where 
		$$f(s):=\sqrt{\frac{ h(\dot \sigma(s) , \dot \sigma (s))}{ 2(e-V(\sigma(s)))}}.$$
		Define the function $\alpha : [a,b] \rightarrow [0,T]$ by $\alpha(s) = \int_a^s f(\tau) d \tau $,  
		where $T=\int_a^b f(\tau) d\tau$. Let $y :[0,T] \rightarrow N$ be defined by $y (s) = \sigma(\alpha^{-1}(s))$, so that $\dot y (s) =\frac{\dot \sigma (\alpha^{-1}(s))}{f(\alpha^{-1}(s))}$.  
		Hence we obtain
		\begin{align*}
			\frac{1}{2}h (\dot{y}, \dot{y}) + V(y)  & = \frac{1}{2} \frac{1}{f(\alpha^{-1})^2}  h (\dot{\sigma}(\alpha^{-1} ) , \dot{\sigma}(\alpha^{-1}  ))  + V(\sigma (\alpha^{-1}))\\
			&  =  (e-V(\sigma (\alpha^{-1}  ))) +  V(\sigma (\alpha^{-1}))= e . 
		\end{align*}
		Moreover, by  \eqref{eul1} 
		\begin{align*}
			\nabla_{\dot{y}}\dot{y} & =  \frac{1}{f(\alpha^{-1})} \cdot \nabla_{\dot{\sigma}(\alpha^{-1})} \left( \frac{1}{f(\alpha^{-1})} \dot{\sigma}(\alpha^{-1}) \right)\\
			& =  \frac{1}{f(\alpha^{-1})}  \left( - f(\alpha^{-1}) \nabla V (\sigma (\alpha^{-1}))  - \big(\iota_{\dot \sigma (\alpha^{-1} )} \mathcal B\big)^\sharp \right)\\
			&  = -  \nabla V (\sigma(\alpha^{-1} )) - \frac{1}{f(\alpha^{-1})} \big(\iota_{\dot \sigma (\alpha^{-1} )} \mathcal B\big)^\sharp   \\
			&= -  \nabla V ( y )- \big(\iota_{\dot y} \mathcal B\big)^\sharp   
		\end{align*}
		thus $y$ solves \eqref{mag}.
		
		Conversely, let   $y:[a,b] \rightarrow N$ be a solution of 
		\begin{align}
			& \nabla_{\dot{y}}\dot{y} + \nabla V(y)   =-\big(\iota_{\dot y} \mathcal B\big)^\sharp \label{fixe1} \\
			& \frac{1}{2}h (\dot{y}, \dot{y}) + V(y) =e . \label{fixe2}
		\end{align} 
		Define $j(s): = 2\big(e-V(y(s))\big)$, $s \in [a,b]$,  let $\alpha : [a,b] \rightarrow [0,T]$ be given by  $\alpha(s) = \int_a^s j(\tau ) d \tau $  where $T=\int_a^b  j(\tau) d \tau$. Let us consider $\sigma : [0,T]\rightarrow N$ defined by  $\sigma(s) = y(\alpha^{-1} (s))$,  so that $\dot{\sigma}(s) = \frac{\dot{y}(\alpha^{-1}(s))}{j(\alpha^{-1} (s))}$.
		By \eqref{fixe2} we obtain
		$$
		h_e (\dot{\sigma}, \dot{\sigma})=  2(e-V(\sigma))h (\dot{\sigma}, \dot{\sigma})= \frac{h( \dot{y}(\alpha^{-1} ), \dot{y}(\alpha^{-1} ))}{j (\alpha^{-1} )} =  \frac{2(e-V(y (\alpha^{-1})))}{j (\alpha^{-1} )}=1. $$
		Thus, by \eqref{jac}, 
		\begin{equation} \label{a}
			\sqrt{ \frac{2 (e-V(\sigma))}{h(\dot \sigma , \dot \sigma)}} =\frac{2 (e-V(\sigma))}{ \sqrt{ h_e(\dot \sigma , \dot \sigma)}}=  2 (e-V(\sigma))
		\end{equation}
		and by  \eqref{fixe1} and \eqref{a}, 
		\begin{align*}
			\nabla_{\dot{\sigma}} \left(\sqrt{ \frac{2 (e-V(\sigma))}{h(\dot \sigma , \dot \sigma)}}\dot{\sigma} \right) &  = \nabla_{\dot{\sigma}} \left(  2 (e-V(\sigma)) \dot{\sigma} \right) = \nabla_{\frac{\dot{y}(\alpha^{-1} )}{j(\alpha^{-1})}} \dot{y}(\alpha^{-1}) \\
			& = \frac{1}{j(\alpha^{-1} ) }\left(- \nabla V (y (\alpha^{-1} )) -\big(\iota_{\dot y(\alpha^{-1})} \mathcal B\big)^\sharp \right) \\
			&= - \frac{\nabla V (\sigma)}{2 (e-V(\sigma))}- \big(\iota_{\dot \sigma} \mathcal B\big)^\sharp \\
			& = - \sqrt{\frac{ h(\dot \sigma , \dot \sigma)}{ 2(e-V(\sigma))}} \nabla V (\sigma) - \big(\iota_{\dot \sigma} \mathcal B\big)^\sharp 
		\end{align*}
		which shows that $\sigma$ solves \eqref{eul}, and therefore  it can be reparametrized to become a geodesic of  $F_e$.
	\end{proof}

\end{document}